\documentclass[11pt]{article} 

\usepackage{vmargin}
\setmarginsrb{1in}{1in}{1in}{1in}{0mm}{0mm}{0mm}{7mm}

\usepackage{booktabs} 

\usepackage[ruled]{algorithm2e} 

\SetAlFnt{\small}
\SetAlCapFnt{\small}
\SetAlCapNameFnt{\small}
\SetAlCapHSkip{0pt}
\IncMargin{-\parindent}

\usepackage[utf8]{inputenc}
\usepackage[czech,english]{babel}
\usepackage[T1]{fontenc}

\usepackage{amsmath, amsthm, amsfonts, amssymb}

\usepackage{enumerate}
	
\newtheorem{thm}{Theorem}[section]
\newtheorem{cor}[thm]{Corollary}
\newtheorem{lem}[thm]{Lemma}

\newtheorem{conj}[thm]{Conjecture}
\newtheorem{ques}[thm]{Question}

\usepackage{tikz}
\tikzstyle{every picture} = [>=latex]

\newtheorem{defn}[thm]{Definition}

\def\ca#1{\mathcal{#1}}
\def\symdiff{\mathop{\bigtriangleup}}

\theoremstyle{definition}

\theoremstyle{remark}
\newtheorem{rem}[thm]{Remark}

\title{A New Perspective on FO Model Checking of Dense Graph Classes\footnote{
Jakub Gajarsk\'y has been supported by the
European Research Council (ERC)~ under the European Union's Horizon 2020 research and innovation programme (ERC Consolidator Grant
DISTRUCT, grant agreement No 648527).
J.~Gajarsk\'y (till 2016), P.~Hlin\v{e}n\'y and J.~Obdr\v{z}\'alek
acknowledge support by the Czech Science Foundation, project no.~14-03501S and currently no.~17-00837S.
Daniel Lokshtanov is supported by Pareto-Optimal Parameterized Algorithms, ERC~ Starting Grant 715744.
M.~S.~Ramanujan acknowledges support from the~Bergen Research Foundation, project~BeHard – Beating Hardness by preprocessing and from Austrian Science Fund (FWF), project P26696 X-TRACT.}
}

\author{
Jakub Gajarsk\'y\thanks{Technical University Berlin. E-mail: {\tt jakub.gajarsky@tu-berlin.de}.}
\and
Petr Hlin\v{e}n\'y\thanks{Masaryk University. E-mail: {\tt hlineny@fi.muni.cz}.}
\and
Daniel Lokshtanov\thanks{University of Bergen. E-mail: {\tt daniello@ii.uib.no}.}
\and 
Jan Obdr\v{z}\'alek\footnotemark[3]
\and
M.S. Ramanujan\thanks{University of Warwick. E-mail: {\tt R.Maadapuzhi-Sridharan@warwick.ac.uk}.}
}

\begin{document}

\maketitle

\begin{abstract}
We study the first-order (FO) model checking problem of dense graphs, namely
those which have FO interpretations in (or are FO transductions of) 
some sparse graph classes.  
We give a structural characterization of the graph classes which
are FO interpretable in graphs of bounded degree. 
This characterization allows us to efficiently compute such an 
FO interpretation for an input graph. 
As a consequence, we obtain an FPT algorithm for successor-invariant 
FO model checking of any graph class which is FO interpretable in
(or an FO transduction of) a graph class of bounded degree. 
The approach we use to obtain these results may also be of independent interest.
\end{abstract}






\section{Introduction}
\label{sec:introduction}
Algorithmic metatheorems are theorems stating that all problems expressible
in a certain logic are efficiently solvable on certain classes of
(relational) structures, e.g.\ on finite graphs. 
Note that the model checking problem for {\em first-order logic} -- 
given a graph $G$ and an FO formula $\phi$ we want to decide whether $G$
satisfies $\phi$ (written as $G\models\phi$) -- is trivially solvable
in time $|V(G)|^{\mathcal{O}(|\phi|)}$.
``Efficient solvability'' hence in this context often means
{\em fixed-parameter tractability} (FPT);
that is, solvability in time $f(|\phi|)\cdot|V(G)|^{\mathcal{O}(1)}$ 
for some computable function $f$.  

In the past two decades
algorithmic metatheorems for FO logic on sparse graph classes
received considerable attention.  After the result of Seese~\cite{Seese96}
establishing fixed-parameter tractability of FO model checking on graphs of
bounded degree there followed a series of results~\cite{FrickG01, DawarGK07,
  DvorakKT10} establishing the same result for increasingly rich sparse
graph classes. This line of research culminated in the result of Grohe,
Kreutzer and Siebertz~\cite{gks14}, who proved that FO model checking is FPT
on nowhere dense graph classes.

The result of Grohe, Kreutzer and Siebertz~\cite{gks14} is essentially the best
possible of its kind, in the following sense:
If a graph class $\ca D$ is monotone (i.e., closed on taking subgraphs)
and not nowhere dense,
then the FO model checking problem on $\ca D$ is as hard as that on all graphs.
Possible ways to continue the research into algorithmic metatheorems for FO
logic include the following two directions:

First, one can study relational structures other than graphs. 
This line of research has recently been initiated by Bova, Ganian and
Szeider~\cite{bgs14}, who gave an FPT algorithm for existential FO model checking 
on partially ordered sets of bounded size of a maximum antichain. Their result
was first improved upon in~\cite{ghoo14} and shortly after that followed the result of
Gajarský et al.~\cite{gajarskyetal15}, who extended~\cite{bgs14} to full FO. 
Apart from these results, very little is known and it remains to be
seen what other types of structures and their parameterizations admit fast
FO model checking algorithms. 

Second, one may consider metatheorems for FO logic on classes of graphs which
are not sparse. 
Again, little is known along this line of research. 
One can mention the result of Ganian et al.~\cite{GHKOST15} 
establishing that certain subclasses of interval graphs
admit an FPT algorithm for FO model checking.
Besides, the aforementioned result of~\cite{gajarskyetal15} 
can also be seen as a result about dense (albeit directed) graphs,
and \cite{gajarskyetal15} actually happens to imply the result of~\cite{GHKOST15}.

\medskip

We would like to initiate a systematic study of dense graph classes for
which the FO model checking problem is efficiently solvable.  It appears
that a natural way to arrive at new graph classes admitting FPT algorithms
for FO model checking, is by {\em means of interpretation, or transduction}.
In a simplified setting of interpretations
-- given a graph $G$ and an FO formula $\psi(x,y)$ with two free
variables, we can define a graph $H=I_\psi(G)$ on the same vertex set as $G$
and the edge set determined by $\psi(x,y)$: a pair of distinct vertices
$u,v$ is an edge of $H$ iff $G \models \psi(u,v)\vee\psi(v,u)$.
We then say that $H$ is interpreted in $G$ using~$\psi$.  
A graph class $\mathcal{D}$ is FO interpretable in a graph class $\mathcal{C}$ if there
exists an FO formula $\psi(x,y)$ such that every member of $\mathcal{D}$ is
interpreted in some member of $\mathcal{C}$ using $\psi$.

For now let us assume we have an efficient FO model checking algorithm for the
previous class $\mathcal{C}$, and
consider the FO model checking problem of the class $\mathcal{D}$.
If an input graph from $\mathcal{D}$ was given together with the corresponding
FO interpretation in a graph from $\mathcal{C}$, then one could easily solve the 
model checking problem using the existing algorithm for $\mathcal{C}$.
This is based on the following natural property of interpretations:
if $H \in \mathcal{D}$ is interpreted in $G \in \mathcal{C}$ 
using formula $\psi(x,y)$,
and our question is to decide whether $H \models \phi$, 
it is a standard routine to construct from $\phi$ and $\psi$ 
a sentence $\phi'$ such that $H \models \phi$ if and only if $G \models \phi'$.
Then $G \models \phi'$ is decided by the algorithm given for~$\mathcal{C}$.

However, if the assumed interpretation (or transduction) is not given, 
then the situation is markedly harder.
In this context we ask the following question:
\begin{ques}
\label{question1}
Let $\mathcal{C}$ be a graph class admitting an FPT algorithm for 
FO model checking, and
$\mathcal{D}$ be a graph class FO interpretable in $\mathcal{C}$.  
Does there exist an FPT algorithm for FO model checking on $\mathcal{D}$?
\end{ques}

As outlined above, the difficulty of this question 
lies in the fact that our inputs come from $\mathcal{D}$, 
without any reference to the respective members of $\mathcal{C}$ in which they are interpreted.
Even if the interpretation formula $\psi(x,y)$ is fixed and known
beforehand, we have generally no efficient way of obtaining the respective
member $G \in \mathcal{C}$ for an input $H \in \mathcal{D}$.
Thus, Question~\ref{question1} can be reduced to the following:

\begin{ques}
\label{question2}
Let $\mathcal{C},\mathcal{D}$ be graph classes such that
$\mathcal{D}$ is FO interpretable in~$\mathcal{C}$.
Does there exist an integer $s$ and a polynomial-time algorithm $\ca A$ such that;
given $H\in \mathcal{D}$ as input, $\ca A$ outputs $G \in \mathcal{C}$
and an FO formula $\psi(x,y)$ of size at most~$s$
such that $H$ is interpreted in $G$ using~$\psi\,$?
\end{ques}

An answer to Question~\ref{question2} is far from being obvious,
and it can strongly depend on the choice of~$\psi$.
Take, for example, the following particular FO interpretation:
A graph $H$ is the {\em square} of a graph $G$ if the edges of $H$ are those
pairs of vertices which are at distance at most $2$ in~$G$.
Then the problem; given $H$ find $G$ such that $H$ is the square of $G$,
is NP-hard~\cite{mms94}.
Another such negative example, specifically tailored to our setting,
is discussed in Section~\ref{sec:hardness}.
These examples show that it is important to choose a suitable
interpretation formula $\psi$ (avoiding the hard cases) 
in an attempt to answer Question~\ref{question2}.

\paragraph{Our contribution} 
We answer both Questions~\ref{question1} and~\ref{question2} in the positive for
the case when $\mathcal{C}$ is a class of graphs of bounded degree.
Our answers cover also the more general case of FO transductions
of bounded-degree classes, and include 
checking successor-invariant FO properties in addition to ordinary FO ones.

We first define near-uniform graph classes (Definition~\ref{def:near-uniform}),
based on a new notion of near-$k$-twin relation, which generalizes the folklore
twin-vertex relation and is related also to the neighbourhood diversity
parameter of~\cite{lam10}.  The idea behind this approach is to classify pairs
of vertices which have almost the same adjacency to the rest of the graph.  The
approach seems promising and may be of independent use in further investigation
of well structured dense graph classes. While the definition of non-uniformity
lends itself well to being used in proofs, it is sometimes unnecessarily technical to
reason about. We therefore also introduce an equivalent notion of near-covered
graph classes (Definition~\ref{def:near-covered}), which is more intuitive,
easier to grasp and offers a slightly different perspective.

We then give an efficient FO model checking algorithm
(Theorem~\ref{thm:MCalgo}) for the near-uniform graph classes. 
This algorithm is based upon the above idea of interpretation;
briefly, given a graph $H$ we use the near-$k$-twin relation for a suitable
value of $k$ to partition the vertex set of $H$
and to find a bounded degree graph $G$, such that $H$ is interpreted in $G$
using a universal formula $\psi$ depending only on the class in question
(Theorem~\ref{thm:decomposition}).
Then we employ the aforementioned algorithm of Seese~\cite{Seese96}.
Furthermore, we extend our algorithm to include also stronger
successor-invariant FO properties (see Section~\ref{subsec:successor} for
more details), for which we can use the recent result of~\cite{hkpqrs17}.

In the second half of the paper we argue that the concept of near-uniform
graph classes is robust and sufficiently rich in content.
We prove that the near-covered (and therefore also near-uniform, since the two are equivalent) graph classes are exactly 
those which are FO interpretable in graphs of bounded degree
(Theorem~\ref{thm:characterization})
and, more generally, that any FO transduction of a graph class of bounded
degree is a near-covered graph class (Theorem~\ref{thm:transductiond}).
The key tool we use is Gaifman's theorem~\cite{Gaifman82}.
At this place we remark that properties of graphs which are FO
interpretable in graphs of bounded degree have already been studied, e.g.,
by Dong, Libkin and Wong
in~\cite{dlw97} in a different context,
but those previous results do not imply our conclusions.

We then complement the previous tractability results with a negative
example of a particular FO interpretation which is NP-hard to ``reverse''
even on the class of graphs of degree at most~$3$ 
(Theorem~\ref{thm:recognizeinthard}).
We finish by sketching some interesting open directions for future research.

\section{Definitions and preliminaries}
\label{sec:interpretation}
We begin by clarifying the terminology and recalling
some established concepts concerning logic on graphs.
We assume that $0$ is a natural number, i.e.~$0\in\mathbb N$.
Let $X\symdiff Y$ denote the symmetric difference of two sets.

\paragraph{Graph theory} 
We work with {\em finite simple undirected graphs}
and use standard graph theoretic notation.
We refer to the vertex set of a graph $G$ as to $V(G)$
and to its edge set as to~$E(G)$.
As it is common in the context of FO logic on graphs,
vertices of our graphs can carry arbitrary labels.

\paragraph{FO logic}
The {\em first-order logic of graphs} (abbreviated as FO) applies the
standard language of first-order logic to a graph $G$ viewed as a relational 
structure with the domain $V(G)$ and the single binary (symmetric) relation $E(G)$.
That is, in FO we have got the standard predicate $x=y$, a binary predicate
$edge(x,y)$ with the meaning $\{x, y\}\in E(G)$, 
an arbitrary number of unary predicates $L(x)$ with the meaning that $x$
holds the label~$L$,
usual logical connectives $\wedge,\vee,\to$, and quantifiers
$\forall x$, $\exists x$ over the vertex set $V(G)$.

For example, $\phi(x,y)\equiv \exists z\big(edge(x,z)\wedge edge(y,z)
	\wedge red(z)\big)$
states that the vertices $x,y$ have a common neighbour in~$G$
which has got label `red'.

\paragraph{Parameterized model checking}
The instances of a parameterized problem can be considered as pairs
$\langle I,k\rangle$ where $I$ is the {main part} of the instance and $k$ is
the \emph{parameter} of the instance; the latter is usually a
non-negative integer.  A parameterized problem is
\emph{fixed parameter tractable (FPT)} if instances $\langle I,k\rangle$ 
of size $n$ can be solved in time $O(f(k)\cdot n^c)$ where $f$ is a 
computable function and $c$ is a constant independent of $k$.
In {\em parameterized model checking}, instances are
considered in the form $\langle(G,\phi),|\phi|\rangle$
where $G$ is a structure, $\phi$ a formula, the question is whether
$G\models\phi$ and the parameter is the size of~$\phi$.

When speaking about the FO model checking problem in this paper, 
we implicitly consider the formula $\phi$ (its size) as a parameter.

\paragraph{Interpretations}
In order to simplify our exposition and proofs we work with
a simplified version of FO interpretations
(note, however, this does not impact generality of our conclusions,
as we will see later).

Let $\psi(x,y)$ be an FO formula with two free variables over the language
of (possibly labelled) graphs such that for any graph and any $u,v$ it holds
that $G \models \psi(u,v) \Leftrightarrow G \models \psi(v,u)$ and
$G \not\models \psi(u,u)$, i.e. the relation on $V(G)$ defined by the
formula is symmetric and irreflexive. From now on we will assume that
formulas with two free variables are symmetric and irreflexive
(which can easily be enforced). 
Given a graph $G$, the formula
$\psi(x,y)$ maps $G$ to a graph $H = I_{\psi}(G)$ defined by $V(H) = V(G)$ and
$E(H) = \{\{u,v\}~|~G \models \psi(u,v) \}$.  We then say that the graph $H$ is
\emph{interpreted} in $G$. Notice that even though the graph $G$ can be
labelled, our graph $H$ is not. 
This is to simplify our notation -- nevertheless, 
one may easily inherit labels from $G$ to $H$ if needed.

In the rest of the paper, whenever we consider
graphs $G$ and $H$ in context of interpretations, graph $G$ will be the
graph in which we are interpreting, and graph $H$ will be the ``result''
of the interpretation.

The notion of interpretation can be extended to graph classes -- to a graph
class $\mathcal{C}$ the formula $\psi(x,y)$ assigns the graph class
$\mathcal{D} = I_{\psi}(\mathcal{C}) = \{H\>|~H=I_{\psi}(G),\, G \in
\mathcal{C}\}$.
We say that a graph class $\mathcal{D}$ is \emph{interpretable} in a graph class
$\mathcal{C}$ if there exists formula $\psi(x,y)$ such that
$\mathcal{D} \subseteq I_{\psi}(\mathcal{C})$.  
Note that
we do not require $\mathcal{D} = I_{\psi}(\mathcal{C})$, as we just want every
graph from $\mathcal{D}$ to have a preimage in $\mathcal{C}$.

Interpretations are useful for defining new graphs from old using logic 
(again, we think of $H$ as a result of application of $\psi$ to $G$),
but can also be used to evaluate formulas on $H$ quickly, provided
that we have a fast algorithm to evaluate formulas on $G$. Let $H = I_{\psi}(G)$,
let $\theta$ be a sentence and let $\theta'$ be a sentence obtained from
$\theta$ by replacing every occurrence of the atom ${edge}(x,y)$ by $\psi(x,y)$. 
Then, obviously, $H \models \theta \Longleftrightarrow G \models \theta'$. 

\paragraph{FO transductions}
While interpretations are restricted in a choice of the target domain
(and, in our case, we even require $V(H) = V(G)$\,),
a more general view is provided by so called transductions,
see Courcelle and Engelfriet~\cite{ce12}.
Informally, in addition to an interpretation this allows
to add to a graph arbitrary ``parameters'' (as labels) and to make
several disjoint copies of the graph.

Here we provide a brief definition based on~\cite{bc10},
simplified to target only the FO graph case.
A {\em basic FO-transduction} $\tau_0$ is a triple $(\chi,\nu,\mu)$
of FO formulas with 0, 1 and 2 free variables, respectively,
such that $\tau_0$ maps a graph $G$ into a graph on the vertex
set $\{v\>|~G\models\nu(v)\}$ and the edge set
$\{\{u,v\}~|~G \models \mu(u,v) \}$ (an induced subgraph of $I_\mu(G)$),
or $\tau_0(G)$ is undefined if $G\not\models\chi$.

The {\em$m$-copy operation} maps a graph $G$ to the graph $G^m$ such that
$V(G^m)=V(G)\times\{1,\dots,m\}$, the subset $V(G)\times\{i\}$ for
each $i=1,2\dots,m$ induces a copy of~$G$ (there are no edges between distinct
copies), and $V(G^m)$ is additionally equipped with a binary relation
$\sim$ and unary relations $Q_1,\dots,Q_m$ such that;
$(u,i)\sim(v,j)$ for $u,v\in V(G)$ iff $u=v$, and $Q_i=\{(v,i): v\in V(G)\}$.
The {\em$p$-parameter expansion} maps a graph $G$ to the set of all
graphs which result by expansion of $V(G)$ by $p$ unary predicates.

Altogether, a many-valued map $\tau$ is an {\em FO transduction}
(of simple undirected graphs) if it is
$\tau=\tau_0\circ\gamma\circ\varepsilon$ where $\tau_0$ is a basic
transduction, $\gamma$ is a $m$-copy operation for some $m$,
and $\varepsilon$ is a $p$-parameter expansion for some~$p$.
Note that, in this formal setting, the formulas of $\tau_0$ 
may also refer to the relations $\sim$ and $Q_i$ established by 
the copy operation~$\gamma$.

We remark, once again, that the result of a transduction $\tau$ of one graph is
generally a set of graphs, due to the involved $p$-parameter expansion.
For a graph class $\mathcal{C}$, the result of a {\em transduction $\tau$ of
the class $\mathcal{C}$} is the union of the particular transduction results,
precisely, $\tau(\mathcal{C}):=\bigcup_{G\in\mathcal{C}}\tau(G)$.


\section{Outline of our approach}

Before diving into technical details of our claims and proofs, we give
a brief exposition of ideas leading to our results. 
We start by explaining the core ideas behind our approach to analysing dense graphs and then
we sketch the how interpretations are combined with our approach to dense graphs to obtain the
results presented in Sections~\ref{sec:near-uniform} and \ref{sec:interpretability}.

\subsection{Locality, indistinguishability, and the new approach}

The existing FPT algorithms for FO model checking of sparse graph classes we
mentioned at the beginning of Section~\ref{sec:introduction} rely heavily on
the use of locality of FO logic -- i.e. the fact that evaluating FO formulas can be
reduced to evaluating \emph{local} FO formulas (cf.  Gaifman's
theorem~\cite{Gaifman82}, also in Section~\ref{sec:interpretability}).  This,
together with the fact that in sparse graphs it is possible to evaluate local
formulas efficiently, made the locality-based approach suitable for studying FO
logic on sparse graphs. The problem with using this approach for dense graphs
is obvious -- in a dense graph the whole graph can be in the 1-neighbourhood of
a single vertex\footnote{This is also true for some sparse graphs, say stars,
  but we hope that it is clear that for dense graphs this can cause substantial
  problems.}. This makes evaluating local formulas around such a vertex
expensive (from the FPT perspective), because this amounts to evaluating them
on the whole graph.

An alternative approach to FO model checking, as described in
Section~\ref{sec:near-uniform}, is based on the concept of
vertex indistinguishability. This approach can be used for dense graphs, but is a bit too limited in its scope.
The key notion here is that of twin vertices -- two vertices of a graph $G$
are \emph{twins} if they have the same neighbourhood. 
The fact that two vertices $u,v$ are twins means that they behave in the same way
with respect to any other vertex in a graph.
Consequently, no FO formula can distinguish between $u$ and $v$.
It is not hard to see that the twin relation is an equivalence on the vertex set of a graph. 
One may also say that the set of vertex neighbourhoods occurring in $G$
is ``covered'' by the set of neighbourhoods of representatives of each twin class of~$G$.
The number of equivalence classes of the twin relation is called 
the \emph{neighbourhood diversity}~\cite{lam10} of a graph,
and graph classes of bounded neighbourhood diversity admit a very simple FPT algorithm
for FO model checking. However, as already mentioned, the problem with this approach is that it is too restrictive --
even such simple graph classes as paths have unbounded neighbourhood diversity. 

Our approach is based on observing that the locality-based approach, when used
on sparse graphs, exploits, in its essence, the indistinguishability of
vertices. Take, for example, the graphs of bounded degree. Here any two
vertices behave the same way with respect to the rest of the vertex set (they
are non-adjacent to it), with only a few exceptions (the vertices in their
neighbourhood). In other words, any two vertices have \emph{almost} the same
neighbourhood. This leads to a relaxation of the notion of twin vertices.  We
say that two vertices are \emph{near-$k$-twins} if their neighbourhoods differ
in at most $k$ vertices. To see how this notion works around the issues with
locality and indistinguishability explained above, let us consider the
near-$k$-twin relation on the class $\mathcal{D}_d$ of graphs of degree at most
$d$ and on the class $\overline{\mathcal{D}_d}$ of its complements. On every
graph from these graph classes, the near-$2d$-twin relation is an equivalence
with just one class.
Yet, graphs from $\overline{\mathcal{D}_d}$ are dense 
and some of them contain universal vertices.

The above considerations lead us to studying graph classes such that for
each graph from these classes there exists a small $k$ such that the
near-$k$-twin relation is an equivalence with a small number of classes --
the \emph{near-uniform} graph classes.  
Though, unlike the ordinary twin relation, the near-$k$-twin relation is not
automatically guaranteed to be an equivalence (this depends heavily on the
choice of $G$ and $k$)
and, consequently, dealing with near-uniformity is slightly cumbersome
and requires a great care.

However, there is also another (and perhaps simpler to deal with) 
way to view and formally capture the above informal discussion of diversity
of neighbourhoods in a graph -- that one can ``cover'' all distinct
neighbourhoods in the graph with only few representative neighbourhoods.
This view leads to a new definition -- a class of graphs is 
\emph{near-covered} if there exists a small $k$ 
such that every graph in this graph class contains a small
(of a constant size) set $S$ of vertices such that every vertex is a
near-$k$-twin of at least one vertex from~$S$.  
It is easily seen that near-uniformity implies near-coveredness 
-- just pick any one representative from each equivalence class.
As we shall see, the converse is also true and the two notions are
(asymptotically) equivalent.
Precisely, we shall prove that for any graph class $\mathcal{C}$ the following
conditions are equivalent:
\begin{enumerate}
\item $\mathcal{C}$ is near-uniform (Definition~\ref{def:near-uniform});
\item $\mathcal{C}$ is near-covered (Definition~\ref{def:near-covered});
\item $\mathcal{C}$ is interpretable in a class of graphs of bounded degree.
\end{enumerate}


Since we can efficiently compute the interpretation claimed in (3),
we can then solve FO model checking on near-uniform graph classes in FPT
using established tools, such as the algorithm of~\cite{Seese96}      
for FO model checking on graphs of bounded degree.
Our proof is structured as follows; 
we first prove the equivalence between (1) and (2) (Lemma~\ref{lem:unifrom-covered}), 
and then the implications (1) $\Rightarrow$ (3) (Theorem~\ref{thm:decomposition}) and (3) $\Rightarrow$ (2) (Theorem~\ref{thm:characterization}).

One may, with respect to technical difficulties related to the
near-uniformity notion, question whether it is necessary to consider
near-uniformity at all and not to go with just near-coveredness alone.
However, the equivalence aspect of the near-$k$-twin relation
is crucial in proving that graphs with certain properties are
interpretable in graphs of bounded degree. We therefore believe that it
deserves a separate definition.

%

\subsection{Interpretability in graphs of bounded degree}
\label{subsec:interp}


Besides dealing with the FO model checking problem via interpretation
of certain graph classes into classes of bounded degree, we are also interested
in the other direction -- to find out which graph classes can be FO interpreted 
into classes of bounded degree (the direction (3) $\Rightarrow$ (2) above).

Our characterization of such classes relies on a
simple corollary of Gaifman's locality theorem: For a~graph~$G$ and two
vertices $u,v \in V(G)$ which are far apart form each other, the truth value of the
formula $\psi(u,v)$ depends only on formulas with one free variable (up to
the quantifier rank $q$, which depends on $\psi$) valid on $u$ and $v$ (i.e. its
logical $q$-types). This in turn means that when the formula $\psi(u,v)$ is
used for interpretation (to obtain the graph $H$ from a graph $G$ of degree at
most $d$) and vertices $u$ and $u'$ satisfy the same formulas with one free
variable (again, up to the quantifier rank $q$), $u$ and $u'$ will be adjacent to the
same vertices in the resulting graph, except for a small number of vertices
which were in their respective $r$-neighbourhoods in graph $G$ (here $r$
also depends on $\psi(x,y)$). Any two vertices of the same $q$-type will
therefore be near-$k$-twins for $k=2\cdot d^r$. 

While the previous consideration is quite simple, note the following
possible pitfall.
Since the relation ``being of the same $q$-type'' is an equivalence with 
a bounded number of classes, it is tempting to believe that 
the near-$k$-twin relation (for a suitably chosen~$k$) 
is an equivalence with a bounded number of classes (independent of $G$)
for any graph from a graph class FO interpretable in a class of graphs of
bounded degree. This, however, is not true -- 
it can happen that some vertices $u$
and $v$ of different $q$-types can be near-$k$-twins and a vertex $w$ of yet
different $q$-type can be near-$k$-twin of $v$ but not of $u$, thus failing the transitivity.

Instead, we finish as follows.
Since for any $q$ there are finitely many (say $m$) different $q$-types, in $H$
there exist at most $m$ vertices such that every
vertex is a near-$k$-twin of (at least) one of them. This in turn means that graph
classes FO interpretable in graphs of bounded degree are near-covered,
and hence also near-uniform by (1) $\Leftrightarrow$ (2) above.


\section{Near-uniform and near-covered graph classes}
\label{sec:near-uniform}

In this section we formally establish the key concepts.
For a graph $G$ and a vertex $v\in V(G)$, 
we define the {\em neighbourhood} of $v$  
as $N^G(v)=\{w\in V(G) \mid \{v,w\}\in E(G)\}$. 
If the graph $G$ is clear from the context, we write just $N(v)$. 
Note that, by definition, $v\not\in N(v)$.

A useful concept in graph theory is that of twin vertices.
Two vertices $u,v\in V(G)$ are called {\em false twins} if
$N(u)=N(v)$, and they are {\em true twins} if
$N(u)\cup\{u\}=N(v)\cup\{v\}$.
We actually follow the concept of false twins, which better suits our
purposes, in the next definition.

\begin{defn}[near-$k$-twin relation]\label{def:nearktwin}
For a graph $G$ and $k\in\mathbb N$, the \emph{near-$k$-twin relation
of~$G$} is the relation $\rho_k$ on $V(G)$ defined
by $(u,v) \in \rho_k$ $\iff$ \mbox{$|N(u) \symdiff N(v)| \le k$}.
\end{defn}

Considering, e.g., $k$ a small parameter and $G$ a large graph
then, intuitively, two vertices of $G$ are near-$k$-twins if they have
``almost the same'' neighbourhood.
This relation, unlike the ordinary {twin relations} on graph vertices, 
does not always ``behave nicely''; in particular,
$\rho_k$ may not be an equivalence relation (see e.g. the examples below).
On the other hand,
if the near-$k$-twin relation is an equivalence of bounded index,
then we can use it to decompose the vertex set of the graph $G$
into similarly behaving clusters.
This leads to the following.

\begin{defn}[near-uniform]\label{def:near-uniform}
A graph $G$ is \emph{$(k_0,p)$-near-uniform} if
there exists $k\leq k_0$ for which near-$k$-twin
relation of $G$ is an equivalence of index at most~$p$.
\\
A graph class $\mathcal{C}$ is \emph{$(k_0,p)$-near-uniform} if every member
of $\mathcal{C}$ is $(k_0,p)$-near-uniform,
and $\mathcal{C}$ is \emph{near-uniform} if there exist integers $k_0,p$ such
that $\mathcal{C}$ is $(k_0,p)$-near-uniform.
\end{defn}

To simplify the discussion, we use the following as a shorthand.
If $\rho_k$ of Definition~\ref{def:nearktwin} is an equivalence relation,
then we call $\rho_k$ the \emph{near-$k$-twin equivalence} of $G$,
and the equivalence classes of  $\rho_k$ the \emph{near-$k$-twin classes
of~$G$.}

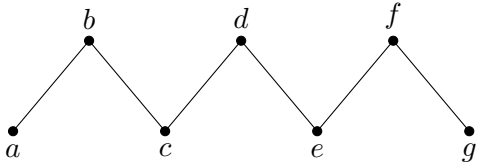
\begin{figure}[t]
$$
\begin{tikzpicture}[yscale=0.6]
\tikzstyle{every node}=[draw, shape=circle, inner sep=1.2pt, fill=black]
\draw (0,0) node[label=below:$a$] {} --
	  (1,2) node[label=above:$b$] {} -- 
	(2,0) node[label=below:$c$] {} --
	  (3,2) node[label=above:$d$] {} -- 
	(4,0) node[label=below:$e$] {} --
	  (5,2) node[label=above:$f$] {} --
	(6,0) node[label=below:$g$] {} ;
\end{tikzpicture}
$$
\caption{An example.
The near-$2$-twin relation $\rho_2$ of this path includes pairs $(b,d)$ and
$(d,f)$ but not $(b,f)$, and so $\rho_2$ is not an equivalence.
On the other hand, $\rho_1$ is an equivalence on this path
and its near-$1$-twin classes are $\{a,c\},\{e,g\},\{b\},\{d\},\{f\}$.}
\label{fig:P6notequiv}
\end{figure}

For example, take a class $\mathcal{D}_d$ of the graphs of maximum degree at
most~$d$, and let $k=2d$.
Then the near-$k$-twin relation $\rho_k$ is a trivial equivalence of index one
(i.e., with one class) for every graph from $\mathcal{D}_d$.
The same holds for the class $\overline{\mathcal{D}}_d$ of the complements of
graphs of~$\mathcal{D}_d$.
Another sort of examples comes, say, with a class $\overline{\mathcal{B}}_d$ of the
graphs obtained from complete bipartite graphs by subtracting a subgraph of
degrees at most~$d$.
For $k=2d$ and every graph of $\overline{\mathcal{B}}_d$, 
the near-$k$-twin relation $\rho_k$ is an equivalence of index at most two.
On the other hand, we can easily see that the near-$2$-twin relation of,
e.g., a path of length~$6$ is not an equivalence; see Figure~\ref{fig:P6notequiv}.
Even more, examples such as that of Figure~\ref{fig:P6notequiv} show
that, having a near-$k$-twin equivalence for some $k$, does not imply that the
near-$k'$-twin relation is an equivalence for $k'>k$.
That is why we cannot simply use one universal value of $k$ in
Definition~\ref{def:near-uniform}.

\medskip 
The fact that the near-$k$-twin relation of a graph $G$ is an equivalence on
$V(G)$ can used as follows: the neighbourhood of a vertex is represented
by the neighbourhood of a selected representative of its class and 
the (small) difference of these two neighbourhoods.
For such purpose of representation it is not always necessary to have
a near-$k$-twin equivalence;
just having \emph{at least} one such representative for every vertex 
of $G$ may be sufficient (we may not care that there are more than one
``close'' representatives).
This simplified scenario leads to the following definition.

\begin{defn}[near-covered]\label{def:near-covered}
A graph $G$ is \emph{$(\ell,q)$-near-covered} if
there exist vertices $v_1, \ldots, v_q$ in $V(G)$ such that each vertex 
$u \in V(G)$ is a near-$\ell$-twin of at least one of $v_1, \ldots, v_q$.
\\
A graph class $\mathcal{C}$ is \emph{$(\ell,q)$-near-covered} if every member
of $\mathcal{C}$ on at least $q$ vertices is $(\ell,q)$-near-covered,
and $\mathcal{C}$ is \emph{near-covered} if there exist integers $\ell,q$ such
that $\mathcal{C}$ is $(\ell,q)$-near-covered.
\end{defn}

The following lemma establishes that the two notions -- being near-uniform and
being near-covered -- are in fact equivalent.  While the definition of being
near-covered is less technical and easier to grasp, the definition of
near-uniformity is more convenient to work with in the algorithmic context of
Section~\ref{sec:algorithm}, which is the main reason for including both
definitions.
\begin{lem}\label{lem:unifrom-covered}
A graph class $\mathcal{C}$ is near-uniform if and only if $\mathcal{C}$ is near-covered.
\end{lem}
\begin{proof}
It is easy to see that if $\mathcal{C}$ is $(k_0,p)$-near-uniform then it is
$(k_0,p)$-near-covered: for every graph $G \in \mathcal{C}$ there is 
$k \le k_0$ such that near $k$-twin relation is an equivalence with
$p'$ classes $C_1, \ldots, C_{p'}$ where $p' \le p$.  We pick an arbitrary 
vertex $v_i$ from from each class $C_i$ to obtain vertices $v_1, \ldots, v_{p'}$. 
Clearly, each vertex of $G$ is a near-$k$-twin of one of these vertices.

To prove the opposite direction, consider first the following construction: 
To any graph $G$ and $k$, we define auxiliary graph $G_k$  on the same vertex set 
by setting $(u,v) \in E(G_k)$ if and only if $u$ and $v$ are near-$k$-twins in $G$. 
Observe the following easy properties of this construction:
\begin{enumerate}
\item Graph $G$ is $(\ell,q)$-near-covered if and only if $G_\ell$ has a dominating set of size at most $q$.
\item\label{item:cliques} If for some $k$ the graph $G_k$ is a disjoint union of at most $p$ cliques, then near-$k$-twin is an equivalence with $p$ classes on $G$ (and so $G$ is $(k,p)$-near-uniform).
\item\label{item:dist} If two vertices are at distance at most $p$ in $G_k$ then they are $pk$-near-twins in $G$.
\item\label{item:component} If $G_k$ contains a component with radius greater than $1$ then this component has to be dominated by at least two vertices. Moreover, in any dominating set of such connected component there are two vertices which are at distance at most $3$ in $G_k$.
\end{enumerate}

We now prove that any graph class $\mathcal{C}$ which is near-covered with
parameters $\ell$ and $q$ is also near-uniform.  We proceed by induction on $q$.  
For the case when $q=1$ the graph $G_\ell$ has a dominating set of size $1$.  
This means that every two vertices in $G_\ell$ are at distance at most $2$, 
it follows from \eqref{item:dist} that any two vertices of $G$ are near-$2\ell$-twins. 
Graph $G$ is therefore $(2\ell,1)$-near-uniform, which finishes the induction basis.

For the induction step, we fix $q>1$ and
assume that every $(m, q-1)$-near-covered graph
class, for any $m$, is $(a,b)$-near-uniform for some values
$a, b$ depending only on $m$ and $q$.
Consider now a graph class $\mathcal{C}$ which is $(\ell,q)$-near-covered.  
We will prove that every graph from
$\mathcal{C}$ is $(2\ell,q)$-near-uniform or $(8\ell,q-1)$-near-covered.
The latter case, from the induction hypothesis, implies that $\mathcal{C}$ is
$(a,b)$-near-uniform where $a,b$ depend only on $\ell$ and~$q$.  
As a result, every graph in $\mathcal{C}$ is
$(\max(2\ell,a), \max(q,b))$-near-uniform 
and so $\mathcal{C}$ is near-uniform.

We take a graph $G\in\mathcal{C}$ which is $(\ell,q)$-near-covered,
and consider the derived graph $G_\ell$. 
If $G_\ell$ has dominating set of size smaller than $q$ then 
it is actually $(\ell,q-1)$-near-covered,
which means it is also $(8\ell,q-1)$-near-covered as desired.  
From now on we therefore assume that $G_\ell$ has a smallest dominating 
set $S =\{v_1, \ldots, v_q\}$. 
We distinguish two cases:

\begin{enumerate}[I.]
\item\label{item:large_radius} 
$G_\ell$ contains a connected component $C$ with radius at least $2$. By
property~\eqref{item:component} of the construction, there are two
vertices $v_i, v_j$ from $S$ which are at distance at most $3$ in $C$. 
Consider now the graph $G_{4\ell}$.  We claim that $G_{4\ell}$ has a dominating
set of size at most $q-1$, which means that $G$ is $(4\ell,q-1)$-near-covered
and therefore also $(8\ell,q-1)$-near-covered as desired.

First note that $G_{4\ell}$ is supergraph of $G_\ell$, so $S$ is a dominating set
of $G_{4\ell}$.  We claim that $S \setminus v_j$ (of size $q-1$) is also a
dominating set of $G_{4\ell}$.  To see this, consider any vertex $u$ dominated
by $v_j$ in $G_\ell$.  Since the distance between $v_i$ and $v_j$ in $G_\ell$ is
at most $3$, the distance between $v_i$ and $u$ is at most $4$ in $G_\ell$. 
This means, by~\eqref{item:dist}, that $v_i$ and
$u$ are $4\ell$-near-twins in $G$.  This in turn means that there is an edge
between $v_i$ and $u$ in $G_{4\ell}$, and so $u$ is dominated by $v_i$ in
$G_{4\ell}$.  Since $u$ was an arbitrary neighbour of $v_j$ (in $G_\ell$), every
vertex dominated by $v_j$ in $G_\ell$ is dominated by $v_i$ in $G_{4\ell}$.
Therefore, $S \setminus v_j$ is a dominating set in $G_{4\ell}$ of size $m-1$.

\item 
All connected components of $G_\ell$ have radius at most $1$.  This means that
$G_\ell$ consists of components $C_1, \ldots, C_q$ such that $v_i \in C_i$ for $i=1,\dots,q$.  
In this case we consider the graph $G_{2\ell}$.  Since every two vertices in the
same component $C_i $ of $G_\ell$ are at distance at most $2$, they are
$2\ell$-near-twins in $G$ and so there is an edge between them in $G_{2\ell}$,
which means that each component $C_i$ forms a clique in $G_{2\ell}$.  We
distinguish two possibilities:

\begin{enumerate}
\item There is no pair of distinct indices $i,j$ such that there exists an edge
in $G_{2\ell}$ between some vertices $u \in C_i$ and $w \in C_j$.
In this case the graph $G_{2\ell}$ is a
disjoint union of $q$ cliques, which means that $G$ is $(2\ell,q)$-near-uniform 
by property~\eqref{item:cliques}.
\item 
There exists a pair of distinct indices $i,j$ such $G_{2\ell}$
contains an edge $uw$ between some vertices $u \in C_i$ and $w \in C_j$.
Recall that $v_i$ and $v_j$ are the vertices from
$S$ which are contained in $C_i$ and $C_j$, respectively.  These vertices
are in the same component in $G_{2\ell}$ and at distance at most $4$.  
By the same argument as in the case~\ref{item:large_radius}, the set $S
\setminus v_j$ is a dominating set of size $q-1$ of the graph $G_{8\ell}$ ,
which means that $G$ is $(8\ell,q-1)$-near-covered, as desired.
\end{enumerate}
\end{enumerate}
\vspace*{-4ex}
\end{proof}

\section{FO model checking algorithm}
\label{sec:algorithm}

This section constitutes the main algorithmic contribution of the paper.

Our model checking algorithm 
for near-uniform graph classes can be shortly summarized as follows.
Input is a graph $H$ from a $(k_0,p)$-near-uniform graph class $\mathcal{C}$ 
and an FO sentence $\phi$.
Perform the following steps:
\begin{enumerate}
\item\label{it:alg-k} For each $k:=0,1,\dots,k_0$;
compute the near-$k$-twin relation $\rho_k$ of $H$,
and check whether $\rho_k$ is an equivalence of index at most $p$.
This test has to succeed for some value of~$k$
(Definition~\ref{def:near-uniform}).

\item Compute a universal formula $\psi(x,y)$ depending on $k_0$ and $p$, 
and the graph $G_H$ depending on $H$ and $k$ found in step~\ref{it:alg-k},
such that $H = I_{\psi}(G_H)$ and the vertex degrees in $G_H$ are at
most~$2k_0p$ (Theorem~\ref{thm:decomposition}).

\item Run the algorithm of~\cite{Seese96}
for FO model checking on graphs of bounded degree on
$G_H$ and the sentence $\phi'$, where $\phi'$ is obtained from $\phi$ 
by replacing every occurrence of $edge(z,z')$ with~$\psi(z,z')$.
\end{enumerate}

\begin{thm}\label{thm:MCalgo}
Let $\mathcal{C}$ be a $(k_0,p)$-near-uniform graph class 
for some \mbox{$k_0,p\in\mathbb N$}.
Then the FO model checking problem of $\mathcal{C}$ is fixed-parameter
tractable when parameterized by the formula size,
i.e., solvable in time $f(|\phi|)\cdot|V(G)|^{\mathcal{O}(1)}$
for a computable function $f$ and input $G,\phi$.
\end{thm}
The rest of this section is devoted to the proof of this statement.

\subsection{Properties of the near-$k$-twin relation}
\label{subsec:near-k-properties}

To give details of the algorithm and to prove Theorem~\ref{thm:MCalgo},
we study some structural properties of graphs for which 
the near-$k$-twin relation is actually an equivalence. 

As outlined above in the algorithm,
our key step is to show that
all near-uniform graph classes are FO interpretable in graph classes of bounded degree.  
For this we show that for any two large enough 
equivalence classes of a near-k-twin equivalence,
it holds that every vertex from one class is connected to almost all or
to almost none vertices of the other class and vice versa.
More precisely:

\begin{lem}
\label{lem:small_degree}
Let $k\geq1$ and $G$ be a graph such that
the near-$k$-twin relation $\rho_k$ of $G$ is an equivalence on~$V(G)$.
Let $U$ and $V$
be two near-$k$-twin classes of $G$  with at least $4k+2$ vertices each
(it may be~$U=V$).
Then for every $v\in V$ we have
\[\min\{|U\cap N(v)|,\> |U\setminus N(v)|\}\leq 2k .\]
\end{lem}

Note that the claim of Lemma~\ref{lem:small_degree} universally holds only
when both $U$ and $V$ are sufficiently large.
A~counterexample with small $U$ is a graph consisting of
$U=\{u\}$ and $V$ inducing a large clique, such that $u$ is connected to
half of the vertices of~$V$.
For this graph the near-$1$-twin classes are exactly $U$ and $V$,
but both $|V\cap N(u)|$ and $|V\setminus N(u)|$ are unbounded.

\begin{proof}
For $x\in V(G)$ and $A\subseteq V(G)$, let 
$\alpha^A(x) = \min\{|N(x)\cap A|,\, |A\setminus N(x)|\}$. Thus to prove the
lemma we need to show that $\alpha^{U}(v) \leq 2k$ for $v\in V$.

Towards a contradiction assume $\alpha^{U}(v) \ge 2k+1$
for some $v\in V$. 
Clearly, there is a subset $U'\subseteq U$ such that
$|U'|=4k+2$ and $\alpha^{U'}(v) \ge 2k+1$, too.
Since $|N(w) \symdiff N(v)| \le k$ for any $w\in V$ by the definition of~$\rho_k$,
we also get $\alpha^{U'}(w)\geq \alpha^{U'}(v)-k \ge 2k+1-k = k+1$
for all $w\in V$.

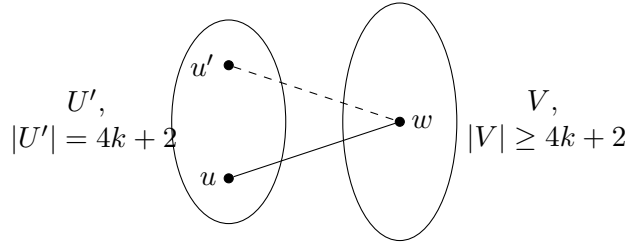
\begin{figure}[t]
$$
\begin{tikzpicture}[scale=0.75]
\tikzstyle{every node}=[draw, shape=circle, inner sep=1.2pt, fill=black]
\draw (2,2) node[label=left:$u$] (u) {};
\draw (5,3) node[label=right:$w$] (w) {};
\draw (2,4) node[label=left:$u'$] (uu) {};
\draw (u) -- (w);
\draw[dashed] (uu) -- (w);
\draw (2,3) ellipse (1cm and 1.8cm);
\draw (5,3) ellipse (1cm and 2.1cm);
\node[fill=none, draw=none, text width=2cm, align=center] at (-0.5,3) {
	$U'$, \mbox{$|U'|=4k+2$} };
\node[fill=none, draw=none, text width=2cm, align=center] at (7.5,3) {
	$V$, \mbox{$|V|\geq4k+2$} };
\end{tikzpicture}
$$
\caption{An illustration; 
counting the pairs $(w,\{u,u'\})$ such that $w\in V$, $u,u'\in U'$
in the proof of Lemma~\ref{lem:small_degree}, in case $U\not=V$.}
\label{fig:counting-uuw}
\end{figure}

We are going to count the number $D$ of pairs $(w,\{u,u'\})$ such that
$w\in V$, $u,u'\in U'$ are distinct vertices
and exactly one of $wu$, $wu'$ is an edge of~$G$.
See Figure~\ref{fig:counting-uuw}.
On the one hand, for any fixed $u,u'\in U'$, every $w$ forming such a
desired pair $(w,\{u,u'\})$ belongs to $N(u)\symdiff N(u')$ 
and so we have got an upper bound
\begin{equation}
  \begin{split}    
  D& \leq \sum_{\{u,u'\} \in {U' \choose 2}} |N(u)\symdiff N(u')| \leq\\
   & \leq {|U'|\choose 2}\cdot k = {4k+2\choose 2}\cdot k 
   < 3k^2(4k+2) \,,
  \end{split}
\label{eq:D-one}
\end{equation}
where $|N(u)\symdiff N(u')|\leq k$ holds by the definition of~$\rho_k$.

On the other hand, we may fix $w\in V$ and count the number of unordered pairs
$u,u'\in U'\setminus\{w\}$ such that exactly one of $wu$, $wu'$ is an edge of~$G$;
this number is equal to $|N(w)\cap U'|\cdot|U'\setminus N(w)|=
 \alpha^{U'}(w)\cdot\big(|U'|-\alpha^{U'}(w)\big)$ if $w\not\in U'$,
and to $\alpha^{U'}(w)\cdot\big(|U'|-1-\alpha^{U'}(w)\big)$ 
or $(\alpha^{U'}(w)-1)\cdot\big(|U'|-\alpha^{U'}(w)\big)$ 
if~$w\in U'$.
Therefore,
\begin{equation}
  \begin{split}    
  D & \geq \sum_{w\in V} \big(\alpha^{U'}(w)-1\big)\cdot
			\big(|U'|-1-\alpha^{U'}(w)\big)\\
   & \geq \sum_{w\in V} (k+1-1)(4k+2-1-k-1)\\
   & =  |V|\cdot3k^2 \geq3k^2(4k+2)
  \end{split}
\label{eq:D-two}
\end{equation}
since we have got $\alpha^{U'}(w) \ge k+1$ and $|V|\geq4k+2=|U'|$.

Now, \eqref{eq:D-one} and \eqref{eq:D-two} are in a contradiction,
and hence the sought conclusion follows.
\end{proof}

\begin{cor}\label{cor:small_degree}
Let $U$ and $V$ be the two classes of Lemma~\ref{lem:small_degree} such that~$|U|,|V|\geq5k+1$.
Then exactly one of the following two possibilities holds:
\begin{enumerate}[(a)]\parskip0pt
\item
every vertex of\/ $U$ is connected to at most $2k$ vertices of~$V$
and every vertex of\/ $V$ is connected to at most $2k$ vertices of~$U$, or
\item\label{it:almostall}
every vertex of\/ $U$ is connected to \emph{all but} at most $2k$ vertices of~$V$
and every vertex of\/ $V$ is connected to \emph{all but} 
at most $2k$ vertices of~$U$.
\end{enumerate}
\end{cor}
\begin{proof}
We first show that either 
\begin{itemize}\parskip0pt
\item every vertex of\/ $U$ is connected to at most $2k$ vertices
  of~$V$, or
\item every vertex of\/ $U$ is connected to all but $2k$ vertices of~$V$.
\end{itemize}
Indeed, for any vertex $v\in U$ taken separately,
only one of these cases can happen since $|V|>4k$,
and one of these cases has to happen by Lemma~\ref{lem:small_degree}.
Assume that there exist $v,w\in U$ with $v$ having at most $2k$ neighbours in $V$
while $w$ is connected to all but at most $2k$ vertices of~$V$.
Then $|N(v) \symdiff N(w)|\geq |V|-2k-2k\geq k+1$,
contradicting the definition of~$\rho_k$.

To finish the proof, we have to show the the following case
(relevant if $U\not=V$) is impossible:
every vertex of\/ $U$ connected to at most $2k$ vertices of~$V$ and
every vertex of\/ $V$ connected to \emph{all but} at most $2k$ vertices of~$U$.
In the argument we count the total number of edges between $U$ and $V$;
it would be at most $2k\cdot|U|$ and, at the same time, at least
$(|U|-2k)\cdot|V|$.
Though, the difference between these lower and upper estimates is
  \begin{equation}
    \begin{split}
      & (|U|-2k)\cdot|V| \,-\, 2k\cdot|U| \\
    =\quad &  |U|\cdot|V| -2k\cdot(|U|+|V|)\\
    =\quad & \big(|U|-4k\big)\big(|V|-4k\big) +2k\big(|U|+|V|\big) -16k^2\\
    >\quad & k\cdot k +2k(5k+5k) -16k^2 = 5k^2>0 ,      
    \end{split}
  \end{equation}
a contradiction, thus finishing the whole proof.
\end{proof}

\begin{rem}\label{rem:single_set}
  Note that Corollary~\ref{cor:small_degree} still applies if $U=V$.
  I.e., for a single
  near-$k$-twin equivalence class $U$ with $|U|>5k+1$ either
  \begin{enumerate}[a)]
  \item every vertex of $U$ has at most $2k$ neighbours in~$U$, or
  \item every vertex of $U$ has at least $|U|-2k$ neighbours in~$U$.
  \end{enumerate}
\end{rem}

\subsection{From near-$k$-twins to bounded degree}

Here we present the core of our algorithm --  a procedure which, given a graph $H$ for which
the near-$k$-twin relation of $H$ is an equivalence of bounded index, produces a
(labelled) graph $G_H$ (on the same vertex set) of bounded degree, and a formula $\psi(x,y)$
such that $H = I_{\psi}(G_H)$. 

The idea behind the procedure is the following: We start by dividing
the near-$k$-twin classes of $H$ into ``small'' and ``large'' ones (w.r.t.~$k$), 
dealing with each of these two types of classes separately. 

\begin{itemize}
\item 
Each large class (more precisely, the vertices in the class) is assigned a label and
each pair of large classes receives another label indicating whether there are
``almost all'' or ``almost none'' edges between the two classes.  The exceptions
to ``almost all'' or ``almost none'' rules will be remembered by edges of
the graph $G_H$ (by Corollary~\ref{cor:small_degree} each vertex has a bounded number of
such exceptions, hence the bounded degree of $G_H$).  Using these labels and
the graph $G_H$ we properly encode the $H$-adjacency between 
the vertices in the large classes.
\item 
The $H$-adjacency of the vertices from small equivalence classes (both within the
small classes and also to the large ones) is encoded by assigning a new
label to each such vertex and another new label to its neighbourhood.  
The vertices from small classes have no edges in the graph $G_H$.
\end{itemize}
 
Note that the construction sketched above depends on $k$ and also on the
number of near-$k$-twin equivalence classes of $H$.  
Unfortunately, as explained earlier, we cannot fix one universal value of
the parameter $k$ beforehand, but at least we can use upper bounds on both
$k$ and the number of equivalence classes (as in Definition~\ref{def:near-uniform}).
With a slightly more complicated use of labels,
we can then give a universal formula $\psi(x,y)$ which depends only on the
parameters $k_0$ and $p$ of a $(k_0,p)$-near-uniform graph
class~$\mathcal{C}$, but is independent from particular $H\in\mathcal{C}$.
This way we get a result even stronger than what is required for the proof of
Theorem~\ref{thm:MCalgo} (see Section~\ref{sec:interpretability} for more
discussion):

\begin{thm}
\label{thm:decomposition}
Let $k_0,p\in\mathbb N$, and $\mathcal{C}$ be a $(k_0,p)$-near-uniform graph class. 
There exists an FO formula $\psi(x,y)$, depending only on $k_0$ and $p$,
such that $\mathcal{C} \subseteq I_{\psi}(\mathcal{D}_{2k_0p})$ 
where $\mathcal{D}_d$ denotes the class of (finite) graphs of degree at most $d$.

Furthermore, for any $H\in\mathcal{C}$ and $k\leq k_0$ such that
the near-$k$-twin relation of $H$ is an equivalence of index at most~$p$,
one can in polynomial time compute a graph $G_H\in\mathcal{D}_{2k_0p}$
such that $H=I_{\psi}(G_H)$.
\end{thm}

\begin{proof}
We are going to prove the theorem by defining the formula $\psi(x,y)$ and,
for each $H \in \mathcal{C}$, efficiently constructing a~graph 
$G_H\in \mathcal{D}_{2k_0p}$ such that $H=I_{\psi}(G_H)$. 
We give the construction of the graph $G_H$ first,
while postponing the definition of $\psi$ to the end of the proof.

Let $0\leq k\leq k_0$ be such that
the near-$k$-twin relation of $H$ is an equivalence of index at most~$p$.
Let $V_1,\ldots,V_m$ where $m \le p$
be the near-$k$-twin classes of $H$ with more than $5k$ vertices
(possible ``small'' near-$k$-twin classes are ignored now).
Observe that $W=V_1\cup\dots\cup V_m$ contains all but at most 
$5k(p-m)\leq5k_0p$ vertices of~$H$.
Let $\overline{W}=V(H)\setminus W$ denote the remaining vertices in
``small'' equivalence classes.
See an illustration in Figure~\ref{fig:classesdivision}.

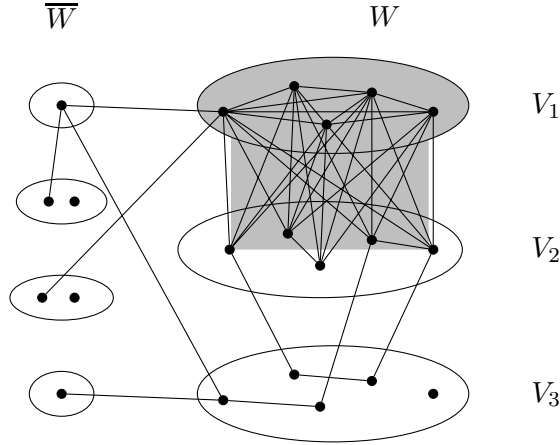
\begin{figure}[t]
$$
\begin{tikzpicture}[scale=0.85]
\tikzstyle{every node}=[draw, shape=circle, inner sep=1.2pt, fill=black]
\draw (1,6) ellipse (0.5cm and 0.35cm);
\draw (1,4.5) ellipse (0.7cm and 0.35cm);
\draw (1,3) ellipse (0.8cm and 0.35cm);
\draw (1,1.5) ellipse (0.5cm and 0.35cm);
\draw (1,6) node (s1) {};
\draw (0.8,4.5) node (s2) {};
\draw (1.2,4.5) node (s22) {};
\draw (0.7,3) node (s3) {};
\draw (1.2,3) node (s33) {};
\draw (1,1.5) node (s4) {};
\draw[color=lightgray, line width=2.6cm] (5.15,5.75) -- (5.15,3.75) ;
\draw[fill=lightgray] (5.2,6) ellipse (2.1cm and 0.75cm);
\draw (5,3.75) ellipse (2.2cm and 0.75cm);
\draw (5.2,1.5) ellipse (2.1cm and 0.75cm);
\draw (6.75,5.9) node (e1) {} ; 
\draw (6.75,3.75) node (e2) {} ; 
\draw (6.75,1.5) node (e3) {} ; 
\node (a1) at (3.5,5.9) {} ; \node (b1) at (4.6,6.3) {} ;
\node (c1) at (5.1,5.7) {} ; \node (d1) at (5.8,6.2) {} ;
\draw (e1) -- (a1) -- (b1) -- (c1) -- (d1) -- (e1) -- (c1) -- (a1) -- (d1) -- (b1) ;
\node (a2) at (3.6,3.75) {} ; \node (b2) at (4.5,4) {} ;
\node (c2) at (5,3.5) {} ; \node (d2) at (5.8,3.9) {} ;
\draw (a1) -- (a2) -- (b1) -- (b2) -- (c1) -- (c2) -- (d1) -- (d2)
	-- (e1) -- (e2) -- (c1) -- (a2) -- (d1) -- (b2) -- (a1)
	-- (d2) -- (e2) -- (a1) ;
\draw (c2) -- (e1) -- (b2) -- (c2) -- (b1) -- (d2) ;
\node (a3) at (3.5,1.4) {} ; \node (b3) at (4.6,1.8) {} ; 
\node (c3) at (5,1.3) {} ; \node (d3) at (5.8,1.7) {} ;
\draw (a3) -- (c3) -- (d2) ; \draw (d1) -- (e2) -- (d3) -- (b3) -- (a2) ;
\draw (s4) -- (a3) -- (s1) -- (a1) ; \draw (s3) -- (a1) ; \draw (s1) -- (s2) ;
\tikzstyle{every node}=[fill=none, draw=none]
\node at (1,7.4) {$\overline W$} ; \node at (6,7.4) {$W$} ;
\node at (8.5,6) {$V_1$} ; \node at (8.5,3.75) {$V_2$} ; \node at (8.5,1.5) {$V_3$} ;
\end{tikzpicture}
$$
\caption{An illustration; small (on the left, $\overline W$) 
and large (on the right, $W$) near-$k$-twin classes of a graph $H$, and
prevailing adjacencies within the large classes remembered by sets $F_1=\{1\}$
and $F_2=\{\{1,2\}\}$,  as in the proof of Theorem~\ref{thm:decomposition}.}
\label{fig:classesdivision}
\end{figure}

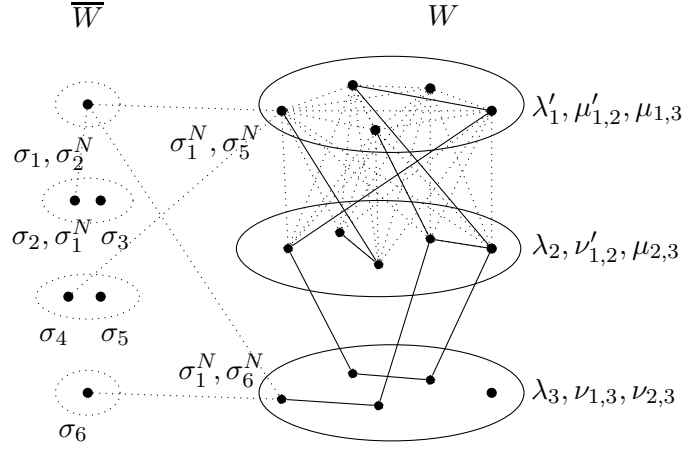
\begin{figure}[t]
$$
\begin{tikzpicture}[scale=0.85]
\tikzstyle{every node}=[draw, shape=circle, inner sep=1.2pt, fill=black]
\draw[dotted] (0.5,6) ellipse (0.5cm and 0.35cm);
\draw[dotted] (0.5,4.5) ellipse (0.7cm and 0.35cm);
\draw[dotted] (0.5,3) ellipse (0.8cm and 0.35cm);
\draw[dotted] (0.5,1.5) ellipse (0.5cm and 0.35cm);
\draw (0.5,6) node[label=below:${\hspace*{-3ex}\sigma_1,\sigma_2^N}~~~$] (s1) {};
\draw (0.3,4.5) node[label=below:${\hspace*{-3.5ex}\sigma_2,\sigma_1^N}$] (s2) {};
\draw (0.7,4.5) node[label=below:$\quad\sigma_3$] (s22) {};
\draw (0.2,3) node[label=below:$\sigma_4\quad$] (s3) {};
\draw (0.7,3) node[label=below:$\quad\sigma_5$] (s33) {};
\draw (0.5,1.5) node[label=below:$\sigma_6\quad$] (s4) {};
\draw (5.2,6) ellipse (2.05cm and 0.75cm);
\draw (5,3.75) ellipse (2.2cm and 0.75cm);
\draw (5.2,1.5) ellipse (2.05cm and 0.75cm);
\draw (6.75,5.9) node[label=right:$\quad{\lambda'_1,\mu'_{1,2},\mu_{1,3}}$] (e1) {} ; 
\draw (6.75,3.75) node[label=right:$\quad{\lambda_2,\nu'_{1,2},\mu_{2,3}}$] (e2) {} ; 
\draw (6.75,1.5) node[label=right:$\quad{\lambda_3,\nu_{1,3},\nu_{2,3}}$] (e3) {} ; 
\node[label=below:$\hspace*{-10.5ex}{\sigma_1^N,\sigma_5^N}$] (a1) at (3.5,5.9) {} ; 
\node (b1) at (4.6,6.3) {} ;
\node (c1) at (4.95,5.6) {} ; \node (d1) at (5.8,6.25) {} ;
\tikzstyle{every path}=[dotted]
\draw (e1) -- (a1) -- (b1) -- (c1) -- (d1) -- (e1) -- (c1) -- (a1) -- (d1) -- (b1) ;
\node (a2) at (3.6,3.75) {} ; \node (b2) at (4.4,4) {} ;
\node (c2) at (5,3.5) {} ; \node (d2) at (5.8,3.9) {} ;
\draw (a1) -- (a2) -- (b1) -- (b2) -- (c1) -- (c2) -- (d1) -- (d2)
	-- (e1) -- (e2) -- (c1) -- (a2) -- (d1) -- (b2) -- (a1)
	-- (d2) -- (e2) -- (a1) ;
\draw (c2) -- (e1) -- (b2) -- (c2) -- (b1) -- (d2) ;
\node[label=above:$\hspace*{-9.5ex}{\sigma_1^N\!,\sigma_6^N}$] (a3) at (3.5,1.4) {} ;
\node (b3) at (4.6,1.8) {} ; \node (c3) at (5,1.3) {} ; \node (d3) at (5.8,1.7) {} ;
\draw (s4) -- (a3) -- (s1) -- (a1) ; \draw (s3) -- (a1) ; \draw (s1) -- (s2) ;
\draw (d1) -- (e2) ;
\tikzstyle{every path}=[solid]
\draw (a3) -- (c3) -- (d2) ; \draw (e2) -- (d3) -- (b3) -- (a2) ;
\draw (b2) -- (c2) ; \draw (d2) -- (e2) -- (b1) ;
\draw (b1) -- (e1) ; \draw (a2) -- (e1) ; \draw (a1) -- (c2) ;
\draw (d2) -- (c1) ;
\tikzstyle{every node}=[fill=none, draw=none]
\node at (0.5,7.4) {$\overline W$} ; \node at (6,7.4) {$W$} ;
\end{tikzpicture}
$$
\caption{An illustration; graph $G_2$ of maximum degree~$3$
constructed for $H$ (the dotted edges) from Figure~\ref{fig:classesdivision},
and the resulting labelling of $V(G_2)=W\cup\overline W$, 
as in the proof of Theorem~\ref{thm:decomposition}.}
\label{fig:classeslabeling}
\end{figure}

We will construct the graph $G_H$ in three stages. First, we define the graph
$G_1=(W, E_1\cup E_2 )$ on the set $W$, where the edge sets are given as:

\begin{itemize}\parskip0pt
\item
Let $F_1$ be the set of those indices $i$ from $\{1,\dots,m\}$ such that
every vertex of $V_i$ has at least $|V_i|-2k$ neighbours in $V_i$ (case~(b) of
Remark~\ref{rem:single_set}). 
We put $E_1=\big\{\{u,v\} \mid u\not=v \land \exists
	 i\in F_1 \text{ s.t. } u,v\in V_i\big\}$.
\item
Let $F_2$ be the set of those index pairs $\{i,j\}$ from $\{1,\dots,m\}$
such that every vertex of\/ $V_i$ is connected to \emph{all but} at most $2k$ vertices of~$V_j$
and every vertex of\/ $V_j$ is connected to \emph{all but} 
at most $2k$ vertices of~$V_i$ (case~(\ref{it:almostall}) of Corollary~\ref{cor:small_degree}). 
We put $E_2=\big\{\{u,v\} \mid \exists \{i,j\}\in F_2 \text{ s.t. }
	 u\in V_i \land v\in V_j \big\}$.
\end{itemize}

In the second step, we adjust $G_1$ by the original edges from~$H$: Let
$E_W=\{\{u,v\}\in E(H) \mid u,v\in W\}$. Then we put $G_2=(W, E(G_1)\symdiff E_W)$.
See in Figure~\ref{fig:classeslabeling}.
Note that every vertex of $G_2$ has degree at most~$2km$ by Corollary~\ref{cor:small_degree}.

In the degenerate case of $k=0$ we arrive at the same conclusion
by the following alternative argument.
By the definition, each near-$0$-twin class is an independent set
and each pair of classes is again independent or induces a complete
bipartite subgraph---this now defines $G_1$ and $G_2$ which is actually edgeless.

In the third step we add back the vertices from $\overline W$ (remember that
$V(H)=W\cup \overline{W}$) by putting $G_H=(W\cup \overline{W}, E(G_2))$. 
Note that $G_H\in \mathcal{D}_{2km}\! \subseteq \mathcal{D}_{2k_0p}$ ,

Finally we label the vertices of $G_H$ by the following fixed label set, 
which is independent of particular~$H\in\mathcal{C}$:
\begin{eqnarray*}
L &:=& \{\lambda_i,\lambda'_i: i=1,\dots,p\} 
\\&&	\cup\, \{\mu_{i,j},\nu_{i,j},\> \mu'_{i,j},\nu'_{i,j}:
		1\leq i<j\leq p \}
\\&&	\cup\, \{\sigma_j,\sigma_j^N: j=1,\dots,5k_0p\} 
\end{eqnarray*}

The vertices of $G_H$ are labelled as follows
(see again Figure~\ref{fig:classeslabeling}):
\begin{itemize}\parskip0pt
\item
For $i=1,\dots,m\leq p$, each vertex of $V_i$ is assigned label 
$\lambda'_i$ if $i\in F_1$, and label $\lambda_i$ otherwise.
\item
For $1\leq i<j\leq m\leq p$, each vertex of $V_i$ is assigned label 
$\mu'_{i,j}$ and each of $V_j$ label $\nu'_{i,j}$ if $\{i,j\}\in F_2$,
and labels $\mu_{i,j}$ and $\nu_{i,j}$, respectively, if $\{i,j\}\not\in F_2$.
\item
Let $\overline{W}=\{w_1,w_2,\dots,w_r\}$ be indexed in any chosen order.
For $j=1,\dots,r\leq 5k_0p$, the vertex $w_j$ is assigned label $\sigma_j$ 
and each neighbour of $w_j$ in $H$ is assigned label $\sigma_j^N$.
\end{itemize}

With $G_H$ in place, we can now define the formula \[\psi(x,y) \>\equiv\>
  (x\neq y)\land(\psi'(x,y)\lor\psi'(y,x))\]
where
\begin{align*}
\psi'(x,y) 
  & \equiv \bigvee_{1\leq \,i\,\leq p} \big(\lambda_i(x)\land\lambda_i(y)\land edge(x,y)\big)\\ 
  & \vee \bigvee_{1\leq \,i\,\leq p}   \big(\lambda'_i(x)\land\lambda'_i(y)\land\neg edge(x,y)\big)\\
  & \vee \bigvee_{1\leq \,i<j\,\leq p}
    \big(\mu_{i,j}(x)\land\nu_{i,j}(y)\land edge(x,y) \big)\\ 
  & \vee \bigvee_{1\leq \,i<j\,\leq p}
    \big( \mu'_{i,j}(x)\land\nu'_{i,j}(y)\land\neg edge(x,y) \big)\\
  & \vee \bigvee_{1\leq \,j\,\leq 5k_0p} 
  \big(\sigma_j(x)\land\sigma_j^N(y) \big)
.\end{align*}

Clearly, $\psi(x,y)$ is independent of particular~$H\in\mathcal{C}$ and
depends only on the parameters~$k_0$ and $p$.
The construction of $G_H$ from $H$ and $k$ is finished in polynomial time and
it is also a simple routine to verify that~$H=I_{\psi}(G_H)$.
\end{proof}

This also finishes the proof of Theorem~\ref{thm:MCalgo}
via the fixed-parameter tractable algorithm of Seese~\cite{Seese96}.

\subsection{Successor-invariant FO}
\label{subsec:successor}

Model checking of successor-invariant FO properties is 
the subject of several recent papers such that 
Engelmann, Kreutzer and Siebertz~\cite{EKS12},
Ganian et al~\cite{GHKOST15},
Eickmeyer and Kawarabayashi~\cite{ek16},
van den Heuvel et al~\cite{hkpqrs17},
and others.

In a nutshell, a {\em successor} relation on a domain $X$ is simply a directed path 
on the vertex set~$X$ (in case of $X$ being the vertex set of a graph,
this successor relation is distinct from the graph edges).
An FO property over a successor-equipped relational structure is 
{\em successor-invariant} if its truth does not change when the same
structure is equipped with a different successor relation.

Since successor-invariant FO sentences are generally more expressive than 
FO sentences~\cite{Ros07}, it makes good sense to ask whether graph classes
with efficient FO model checking algorithms also admit efficient
successor-invariant FO model checking.
So far, the answers provided by the previously listed works are all positive.
Furthermore, since the known examples separating the expressive powers of plain
FO and successor-invariant FO are dense (containing large cliques),
and the previous works (except~\cite{GHKOST15}) studied sparse graphs,
it is especially relevant to ask the question of
successor-invariant FO in our dense case.

The answer here is again positive and plain easy, in fact, 
the following directly follows from our
Theorem~\ref{thm:decomposition} and the algorithm of~\cite{hkpqrs17}:

\begin{cor}
Let $\mathcal{C}$ be a $(k_0,p)$-near-uniform graph class 
for some \mbox{$k_0,p\in\mathbb N$}.
Then the successor-invariant FO model checking problem in $\mathcal{C}$ 
is fixed-parameter tractable when parameterized by the formula size.
\end{cor}
\begin{proof}
We proceed in the same way as previously.
To recapitulate, let $\phi$ be a successor-invariant FO sentence 
and $H\in\mathcal{C}$ an input graph.
By Theorem~\ref{thm:decomposition}, we get formula $\psi$ and compute
a graph $G_H\in\mathcal{D}_{2k_0p}$ such that $H=I_{\psi}(G_H)$.
Let $\phi'$ be obtained from $\phi$
by replacing every occurrence of $edge(z,z')$ with~$\psi(z,z')$.
Now, it is important that the domain of $G_H$ is the same as the domain
of~$H$, and so any successor relation on $G_H$ is a successor relation
on~$H$ as well.
Consequently, for any successor-equipped $H$ and $G_H$ we have
$H \models \phi$ if and only if $G_H \models \phi'$
(as with previous plain FO logic).

Hence it remains to solve in FPT the successor-invariant FO model checking
problem of graphs of bounded degree.
This can be done by the algorithm of~\cite{hkpqrs17}
(which handles more generally graph classes with bounded expansion).
\end{proof}

\section{Interpretability of graphs of bounded degree}
\label{sec:interpretability}

Having defined near-uniform graph classes and shown that these classes can be
FO interpreted in graph classes of bounded degree, it is a natural question
to ask what is the exact relationship between those kinds of classes. As it
turns out, we can prove (Theorem~\ref{thm:characterization}) that each class
FO interpretable in a class of graphs of bounded degree is indeed
near-covered and therefore also near-uniform. 
Thus, near-uniform graph classes are exactly those graph
classes which are FO interpretable in graph classes of bounded degree.
This result can then be easily extended to the more general case of
transductions of graph classes of bounded degree, which again result in
near-uniform graph classes (Theorem~\ref{thm:transductiond}).

\subsection{Adjusted Gaifman's theorem}

In the proof of the main result of this section we use the famous
Gaifman's locality
theorem~\cite{Gaifman82} (see also~\cite{lib04}) about the local nature of the FO logic. However, for our purposes we
need a specific variant of this theorem. To keep the paper
self-contained, in this section we first recap the notation and statement of
Gaifman's theorem and then state and prove a corollary tailored to our needs. 

An FO formula $\phi(x_1, \ldots, x_l)$ is $r$-\emph{local}, sometimes denoted by
$\phi^{(r)}(x_1, \ldots, x_l)$, if for every graph $G$ and all $v_1, \ldots,
v_l \in V(G)$ it holds
$G \models \phi(v_1, \ldots, v_l) \Longleftrightarrow \bigcup_{1 \le i \le l}
N_r^{G}(v_i) \models \phi(v_1, \ldots, v_l)$, where $N_r^{G}(v)$ is the subgraph
of $G$ induced by $v$ and all vertices of distance at most $r$ from $v$.

\begin{thm}[Gaifman's locality theorem]\label{thm:Gaifman}
Every first-order formula with free variables $x_1, \ldots, x_l$ is equivalent to a
Boolean combination of the following
\begin{itemize}
\item local formulas $\phi^{(r)}(x_1, \ldots, x_l)$ around $x_1, \ldots, x_l$,
and
\item basic local sentences, i.e. sentences of the form
\end{itemize}
$$ 
	\exists x_1 \ldots \exists x_k	\left(\bigwedge_{1 \le i < j \le k} 
		dist(x_i,x_j) > 2r 
	\land \bigwedge_{1 \le i \le k} \phi^{(r)}(x_i) \right)
.$$
\end{thm}

For a given $q$, the set of semantically different FO formulas $\phi$ of
quantifier rank $qr(\phi)\leq q$ with one free variable is finite. Clearly, this also holds for local FO formulas, as they are a special case of FO formulas.  
For a vertex $v$ of a graph $G$, we define its local logical FO $(\rho,r)$-type as
$\text{tp}^G_{\rho,r}(v) = \{\phi^{(r)}(x)~|~G \models \phi^{(r)}(v) \text{ and
}  qr(\phi) \le \rho \}$.

It can be derived from Gaifman's theorem that if two vertices $u$ and $v$
are far apart in the graph, then whether $\psi(u,v)$ holds true depends only
on the logical $(\rho,r)$-type of $u$ and $v$, where $q$ and $r$ depend on~$\psi$. 
This finding is formalized by the following (folklore) corollary of
Theorem~\ref{thm:Gaifman}; as we were not able to find this precise formulation in
the literature, we also provide a proof, for the sake of completeness.

\begin{cor}
\label{cor:our_gaifman}
For every FO formula $\psi(y,z)$ of two free variables
there exist integers $r$ and $\rho$ such that the following holds true
for any graph $G$:
If $u,v_1,v_2 \in V(G)$ such that
$\text{dist}(u,v_1)>2r$, $\text{dist}(u,v_2)>2r$ and\/ 
{\rm$\text{tp}^G_{\rho,r}(v_1) = \text{tp}^G_{\rho,r}(v_2)$},
then \mbox{$G \models \psi(u,v_1)$} if and only if $G \models \psi(u,v_2)$.
\end{cor}

\begin{proof}
Let $\psi(y,z)$ be a formula, $G$ a graph and $u,v_1,v_2 \in V(G)$ as in the
statement of the Corollary.  By Theorem~\ref{thm:Gaifman}, $\psi(y,z)$ is equivalent
to a Boolean combination of local formulas $\phi^{(r)}(y,z)$ around $y$ and
$z$ and basic local sentences.  The validity of $\psi(y,z)$ for any choice
of $y$ and $z$ therefore depends only on local formulas $\phi^{(r)}(y,z)$
around $y$ and $z$ (because the validity of basic local sentences is
independent of the choice of $y$ and $z$).  Thus, whether $G \models
\psi(u,v_1)$ holds true depends only on formulas $\phi^{(r)}(u, v_1)$ evaluated
on the graph induced by $N_r^{G}(u) \cup N_r^{G}(v_1)$.

Because $\text{dist}(u,v_1)>2r$, this graph is actually a disjoint union of
the graphs induced by $N_r^{G}(u)$ and $N_r^{G}(v_1)$.  By the standard
Ehrenfeucht-Fraisse games argument, validity of $N_r^{G}(u) \cup N_r^{G}(v_1)
\models \phi^{(r)}(u,v_1)$ is then fully determined by the types
$\text{tp}^G_{\rho,r}(u)$ and $\text{tp}^G_{\rho,r}(v_1)$ of $u$ and $v_1$
respectively.  The same reasoning can be applied to $u$ and $v_2$, and since
$\text{tp}^G_{\rho,r}(v_1) = \text{tp}^G_{\rho,r}(v_2)$, the result follows.
\end{proof}

\subsection{Characterization of interpretations}

The following theorem provides us with a strong characterization of the
classes FO interpreted in graphs of degree at most~$d$,
in terms of near-$k$-twin relation and being near-covered. It amounts to, in
an essence, the ``opposite direction'' to Theorem~\ref{thm:decomposition}.

\begin{thm}
\label{thm:characterization}
Let $\mathcal{D}_d$ be the class of (finite) graphs with maximum degree at
most $d$ and let $\psi(x,y)$ be an FO formula with two free variables.  
Then there exist $\ell$ and $q$, depending only on $d$ and $\psi$, such that
every graph $H \in I_{\psi}(\mathcal{D}_d)$ is {$(\ell,q)$-near-covered}.
That is, there exists a set $S=\{v_1, \ldots, v_q\} \subseteq V(H)$ 
such that every $u \in V(H)$ is a near-$\ell$-twin of at least one element $v_i$ of~$S$.
\end{thm}


\begin{proof}
Let $H \in I_{\psi}(\mathcal{D}_d)$ and $G\in\mathcal{D}_d$ be such that $H=I_{\psi}(G)$.
Recall that $V(H)=V(G)$ and $\{u,v\}\in E(H)$ if and only if $G\models\psi(u,v)$.
Fixing $G$ and $H$, we say that a vertex $x\in V(H)$ is {\em$a$-far from $y\in V(H)$}
if the graph distance from $x$ to $y$ in $G$ is greater than~$a$.

Let $\rho$ and $r$ be values obtained from application of
Corollary~\ref{cor:our_gaifman} to $\psi$.  Let $T$ be the set of all
possible FO $(\rho,r)$-types (this set is finite for any $\rho$ and $r$) and let
$T_G\subseteq T$ be the set of all FO $(\rho,r)$-types realized in $G$.  
Set $q: = |T_G|$ and $\ell:= 2d^r$.  For every $t \in T_G$ we pick one vertex
in $G$ (and also in $H$ since $V(H)= V(G)$) which realizes type $t$ 
to obtain $S:= \{v_1,\ldots, v_q\}$.  
We claim that every vertex $u$ of $H$ is near-$\ell$-twin of some
vertex in~$S$.  For every vertex $u$ of $H$, there is a
vertex $v_i$ in $S$ with the same $(\rho,r)$-type (in $G$).  By
Corollary~\ref{cor:our_gaifman}, for every $w$ which is more than $r$ far
from both $u$ and $v_i$, it holds that $G \models \psi(w,u)$ if and only if
$G \models \psi(w,v_i)$.  This means that $u$ and $v_i$ will have the same
adjacency to all vertices which are more than $r$ far from both of them. 
Consequently, their neighbourhoods in $H$ can only differ in vertices which
are at most $r$ far from either of them, and there are at most $2d^r$ such
vertices.
\end{proof}

\subsection{Characterization of transductions}
\label{subsec:transduction}

Besides the simplified case of interpretation from
Theorem~\ref{thm:characterization}, we can go much further with a bit of
additional effort, as carried out in the following claim.
\begin{thm}
\label{thm:transductiond}
Let $\mathcal{D}_d$ be the class of (finite) graphs with maximum degree at
most $d$ and let $\tau$ be an FO transduction.
Then there exist $\ell$ and $q$, depending only on $d$ and $\tau$, such that
every graph $H \in \tau(\mathcal{D}_d)$ is {$(\ell,q)$-near-covered}.
\end{thm}

\begin{proof}
Our strategy is to prove that there exist an integer $d'$,
an FO formula $\psi(x,y)$ and a class of labelled graphs 
$\mathcal{G}\subseteq\mathcal{D}_{d'}$, 
all depending on $d$ and $\tau$, such that the following holds:
for every $H \in \tau(\mathcal{D}_d)$ there exists a graph
$H'\in I_\psi(\mathcal{G})$
such that $H'$ is obtained from $H$ by adding isolated vertices.

Assuming the previous for a moment, we show how it implies our theorem.
By Theorem~\ref{thm:characterization}, there exist $\ell_1$ and $q$
such that, for every graph $H'\in I_\psi(\mathcal{G})$, the following holds:
there exists $S_1=\{v_1, \ldots, v_q\} \subseteq V(H')$ such that 
every $u \in V(H')$ is a near-$\ell_1$-twin of at least one element of~$S_1$.
We may assume that at most one element of $S_1$, say $v_1$,
is among the added isolated vertices of $H'$
(hence $S_1\setminus V(H)\subseteq\{v_1\}$).
If $v_1\not\in V(H)$ and there exists a vertex $w_1\in V(H)$ which is 
a near-$\ell_1$-twin of~$v_1$,
then $w_1$ is a near-$2\ell_1$-twin of every near-$\ell_1$-twin of~$v_1$,
and so $S:=(S_1\cap V(H))\cup\{w_1\}$ witnesses that $H$ is
$(2\ell_1,q)$-near-covered.
Otherwise, $H$ is $(\ell_1,q)$-near-covered by~$S:=S_1\cap V(H)$.


\smallskip
It remains to define an appropriate class $\mathcal{G}$ and
the formula $\psi$ as claimed above.

From the definition of FO transduction, let
$\tau=\tau_0\circ\gamma\circ\varepsilon$ where $\tau_0$ is a basic
transduction, $\gamma$ is a $m$-copy operation for some $m$,
and $\varepsilon$ is a $p$-parameter expansion for some~$p$.
We start with setting $d'=\max(d+1,m)$, 
and $\mathcal{G}_1=\varepsilon(\mathcal{D}_d)$.
For every $G_1\in\mathcal{G}_1$, we take $G_2=\gamma(G_1)$.
Then, for every $v\in V(G_1)$, we add to $G_2$ a new vertex $v_0$
of a new label~$R$ (the same for each added $v_0$),
and $m$ edges from $v_0$ to the $m$ copies of $v$ in~$G_2$
(making a star $K_{1,m}$ with the new centre $v_0$).
The resulting graph $G_3$ has $(m+1)\cdot|V(G)|$ vertices and maximum
degree~$d'$.
Finally, after formally erasing the relation $\sim$ of $\gamma$, 
we add $G_3$ to~$\mathcal{G}$.

Regarding the formula $\psi(x,y)$, we recall that the basic transduction 
$\tau_0$ underlying $\tau$ is determined by a triple of FO formulas $(\chi,\nu,\mu)$, 
where the role of $\chi$ can be safely ignored for now.
We make formulas $\nu'(x)$ from $\nu(x)$, and $\mu'(x,y)$ from $\mu(x,y)$,
by restricting every quantifier to the vertices
{\em not} of label $R$ and replacing each occurrence of the predicate
$u\sim v$ (recall that $\gamma$ has been erased from~$G_3$)
with $\exists t\big(R(t)\wedge edge(u,t)\wedge edge(v,t)\big)$.
Then we set
\begin{equation}\label{eq:psitrans}
\psi(x,y) \equiv \neg R(x)\land\nu'(x) \land
	\neg R(y)\land\nu'(y) \land \mu'(x,y)
\,.\end{equation}

Pick now any $H \in \tau(\mathcal{D}_d)$, and let 
$G_1\in\mathcal{G}_1=\varepsilon(\mathcal{D}_d)$ be the
corresponding graph such that $G_2=\gamma(G_1)$ and $H=\tau_0(G_2)$.
Let $G_3\in\mathcal{G}$ be constructed from $G_2$ as above.
By $\tau_0$ and \eqref{eq:psitrans}, every vertex of $I_{\psi}(G_3)$
not in $V(H)$ is isolated in $I_{\psi}(G_3)$.
Moreover, by the construction of $\psi'$
 for every two vertices $u,v\in V(H)$ it holds
$G_2\models\mu(u,v)$ $\iff$ $G_3\models\mu'(u,v)$.
Hence, by \eqref{eq:psitrans}, $I_{\psi}(G_3)$ results from $H$ by adding
isolated vertices and we can set $H'=I_{\psi}(G_3)$, as desired.
The proof is finished.
\end{proof}

%

Putting together the results of Theorems~\ref{thm:decomposition}
and~\ref{thm:transductiond}, we easily get also the following corollary
which is interesting on its own:

\begin{cor}\label{cor:robust}
Let $\mathcal{C}$ be a near-uniform graph class, and $\tau$ be an FO
transduction.
Then the class $\tau(\mathcal{C})$ is again a near-uniform graph class.
\end{cor}

\begin{proof}
By Theorem~\ref{thm:decomposition},
there exists an FO formula $\psi(x,y)$
such that $\mathcal{C} \subseteq I_{\psi}(\mathcal{D}_{d})$
for suitable degree bound $d$ depending on~$\mathcal{C}$.
Let $\tau_1$ be the corresponding basic transduction 
(determined by $(true,true,\psi)\,$).
Then $\tau(\mathcal{C})\subseteq \tau(\tau_1(\mathcal{D}_{d}))$, and since
$\tau\circ\tau_1$ is again a transduction by transitivity,
by Theorem~\ref{thm:transductiond} every graph of
$\tau(\tau_1(\mathcal{D}_{d}))$ is $(\ell,q)$-near-covered.
Consequently, using Lemma~\ref{lem:unifrom-covered}, 
$\tau(\mathcal{C})$ is also a near-uniform graph class.
\end{proof}

\section{Hardness of recognizing an interpretation}
\label{sec:hardness}

Recall the aforementioned result \cite{mms94}
claiming that it is NP-hard to decide whether a given graph is a 
{\em square} of some graph.
The square of a graph can be straightforwardly described by an FO
interpretation with $\psi_s(x,y)\equiv edge(x,y)\vee[x\not=y\wedge\exists z
 (edge(x,z)\wedge edge(z,y))]$, 
expressing that edges of the square are original edges or pairs at distance
exactly two.

In our context, \cite{mms94} hence means that there
exist a graph class $\mathcal{C}$ and an FO formula $\psi(x,y)$ such that
the problem, for a given graph $H\in I_{\psi}(\mathcal{C})$,
to find $G\in\mathcal{C}$ such that $H=I_{\psi}(G)$ is not efficiently solvable
(unless P=NP).
Though, the reduction of \cite{mms94} requires
a class $\mathcal{C}$ of unbounded maximum degree while we are primarily
interested in interpretations of the classes $\mathcal{D}_d$ of graphs of degrees
at most~$d$.
Here we show a straightforward alternative reduction working already with the class
$\mathcal{D}_3$ of graphs of degree at most~$3$.

Notice that such a result is not in a contradiction with
Theorem~\ref{thm:decomposition} since each of the two results speaks about
a different particular formula(s)~$\psi$.

\begin{thm}\label{thm:recognizeinthard}
Let $\mathcal{D}_3$ denote the class of graphs of degree at most~$3$.
There exists an FO formula $\psi_0(x,y)$ such that
the problem, for a given graph $H\in I_{\psi_0}(\mathcal{D}_3)$,
to find $G\in\mathcal{D}_3$ such that $H=I_{\psi_0}(G)$ is NP-hard.
\end{thm}

\begin{proof}
We reduce from the folklore NP-hard problem of $3$-colouring
a given $4$-regular graph~$H_0$.
We construct a graph $H$ from an arbitrary $4$-regular graph $H_0$ as follows:
\begin{itemize}
\item Every vertex $v$ of $H_0$ is replaced with a graph $T_v$ which is a
copy of the graph in Figure~\ref{fig:vertexgadget} including the dashed edges.
\item Every edge $e$ of $H_0$ is replaced with a graph $U_e$ which is a
copy of the graph in Figure~\ref{fig:edgegadget} including the dashed edges.
\item For every edge $e=\{u,v\}$ of $H_0$, the terminal $e^1$ of $U_e$ is
identified with $u^i$ of $T_u$, and $e^2$ of $U_e$ is identified with $v^j$ of $T_v$,
where $e$ is the $i$-th edge at $u$ and the $j$-th edge at~$v$
(for arbitrarily chosen orderings of edges incident to $u,v$).
\end{itemize}
The construction of $H$ is independent of whether $H_0$ is $3$-colourable.
Note that since $U_e$ contains a vertex of degree~$5$, it is
$H\not\in\mathcal{D}_3$.

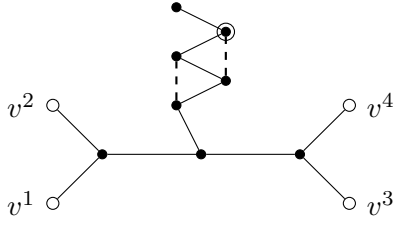
\begin{figure}[t]
$$
\begin{tikzpicture}[scale=0.65]
\tikzstyle{every node}=[draw, shape=circle, inner sep=1.6pt, fill=white]
\node[label=left:$v^1$] (v1) at (0,0) {} ; 
\node[label=left:$v^2$] (v2) at (0,2) {} ; 
\node[label=right:$v^3$] (v3) at (6,0) {} ; 
\node[label=right:$v^4$] (v4) at (6,2) {} ; 
\tikzstyle{every node}=[draw, shape=circle, inner sep=1.2pt, fill=black]
\node (w) at (1,1) {} ; \node (ww) at (5,1) {} ;
\draw (v1) -- (w) -- (v2) ; \draw (v3) -- (ww) -- (v4) ;
\node (t0) at (3,1) {} ;
\draw (w) -- (t0) -- (ww) ;
\node (t1) at (2.5,2) {} ; \node (t2) at (3.5,2.5) {} ;
\node (t3) at (2.5,3) {} ; \node (t4) at (3.5,3.5) {} ;
\node (t5) at (2.5,4) {} ;\draw (t4) circle (5pt) ;
\draw (t0) -- (t1) -- (t2) -- (t3) -- (t4) -- (t5) ;
\draw[dashed, thick] (t1) --(t3) ;
\draw[dashed, thick] (t2) --(t4) ;
\end{tikzpicture}
$$
\caption{The vertex gadget $T_v$ in the proof of
Theorem~\ref{thm:recognizeinthard}.}
\label{fig:vertexgadget}
\end{figure}
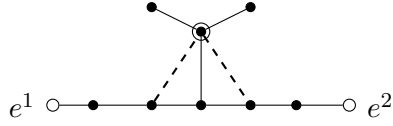
\begin{figure}[t]
$$
\begin{tikzpicture}[scale=0.65]
\tikzstyle{every node}=[draw, shape=circle, inner sep=1.6pt, fill=white]
\node[label=left:$e^1$] (e1) at (0,0) {} ; 
\node[label=right:$e^2$] (e2) at (6,0) {} ; 
\tikzstyle{every node}=[draw, shape=circle, inner sep=1.2pt, fill=black]
\node (c) at (3,0) {} ;
\node (d1) at (2,0) {} ; \node (d2) at (4,0) {} ;
\draw (e1) -- node (x) {} (1.5,0) -- (d1) -- (c)
	 -- (d2) -- node (x) {} (5.75,0) -- (e2) ; 
\node (c1) at (3,1.5) {} ; \draw (c1) circle (5pt) ;
\node (c2) at (2,2) {} ; \node (c3) at (4,2) {} ;
\draw (c) -- (c1) -- (c2) ; \draw (c1) -- (c3) ;
\draw[dashed, thick] (d1) --(c1) -- (d2) ;
\end{tikzpicture}
$$
\caption{The edge gadget $U_e$ in the proof of
Theorem~\ref{thm:recognizeinthard}.}
\label{fig:edgegadget}
\end{figure}

Before defining the formula $\psi_0$, we briefly explain the underlying idea
of the reduction.
For a suitable subgraph $G$ of  $H$ (on the same vertex set),
we would like to have $H=I_{\psi_0}(G)$ if and only if
every vertex gadget (of a vertex of $H_0$) restricted to $G$ encodes one of three
available colours (for this vertex in~$H_0$), 
and every edge gadget in $G$ ``verifies'' that the ends
of the edge (in $H_0$) receive distinct colours.

The above rough sketch is made precise now.
Considering colours $1,2,3$, we define three {\em reduced vertex gadgets}
of a vertex $v\in V(H_0)$
as $T_v^1=T_v$ and $T_v^2,T_v^3$ obtained from $T_v$ by removing one 
or the other dashed edge of~$T_v$ in Figure~\ref{fig:vertexgadget}.
Similarly, a {\em reduced edge gadget} $U_e'$ of an edge $e\in E(H_0)$
is obtained from $U_e$ in Figure~\ref{fig:edgegadget} by removing both dashed edges.
Assuming any $3$-colouring $c:V(H_0)\to\{1,2,3\}$,
we construct a graph $G\in\mathcal{D}_3$ analogously to the above construction of
$H$, while replacing every vertex $v\in V(H_0)$ with $T_v^{c(v)}$\
and every edge  $e\in E(H_0)$ with $U_e'$.

Note that $G\subset H$.
We call a vertex $w$ a {\em v-marker} if $w$ is adjacent to precisely
one vertex of degree $1$, and we call $w$ an {\em e-marker} if $w$ is
adjacent to two vertices of degree $1$
(see the circled vertices in Figures~\ref{fig:vertexgadget}
and~\ref{fig:edgegadget}, respectively).
Then every e-marker $w$ of~$G$ belongs to some $U_e'$ of $e=\{u,v\}\in E(H_0)$,
and there are precisely two v-markers of~$G$ at distance $9$ from~$w$
belonging to $T^i_u$ and to $T^j_v$.
We would now like to ``verify'' that the colouring $c$ is proper, i.e. that~$i\not=j$,
in the formula~$\psi_0$.

We define $\psi_0(x,y)\equiv edge(x,y)\vee \nu(x,y)\vee \eta(x,y)$ where
\begin{itemize}
\item $\nu(x,y)$ asserts that there exists $z$ which is a neighbour of $x$
or~$y$, such that $z$ is a v-marker and the $5$-neighbourhood of $z$ is
isomorphic to one of $T_v^1,T_v^2,T_v^3$, and that $x,y$ are the ends of one
of the dashed edges in Figure~\ref{fig:vertexgadget};
\item $\eta(x,y)$ asserts that one of $x,y$, say $x$, is an e-marker, 
$y$ is at distance two from $x$, and the following holds:
there exist vertices $z,z'$ at distance $9$ from $x$ such that $z,z'$
are v-markers with their $5$-neighbourhoods isomorphic to $T_v^i$
and $T_v^j$ where~$i\not=j$.
\end{itemize}
It is routine to rewrite the above description into an FO formula.

Clearly, $H=I_{\psi_0}(G)$ if and only if the above colouring $c$ is proper.
Conversely, it remains to prove that if $H=I_{\psi_0}(G)$ for any
$G\in\mathcal{D}_3$, then $H_0$ is $3$-colourable.
Notice that $G\subseteq H$ and that the formula $\psi_0$ does not ``add''
edges to degree-$1$ vertices, and so the degree-$1$ vertices of $G$ must be
in a one-to-one correspondence with the v-marker and e-marker vertices of~$H$.

Fix an e-marker $w$ belonging to $U_e\subseteq H$.
Since $w$ is of degree $5$ in $H$ and of degree $\leq3$ in $G\in\mathcal{D}_3$,
it is $G\models\psi_0(w,t)$ for some (actually, at least two) neighbour $t$ of~$w$ in~$H$.
In particular, by the definition of $\eta(w,t)$, this means there exist 
two v-markers $w',w''$ at distance $9$ from~$w$ in~$G$.
From the construction of $H$ we know that $w',w''$ belong to $T_u,T_v$,
respectively, where $u,v$ are the ends of $e$ in $H_0$.
Again by $G\models\psi_0(w,t)$, the subgraph of $G$ induced by $V(T_u)$
is one of $T_u^1,T_u^2,T_u^3$, say it is $T_u^i$.
Similarly, the subgraph of $G$ induced by $V(T_v)$ is, say, $T_v^j$
and $i\not=j$.
Since the same holds for any edge of $H_0$, an (arbitrary) graph
$G\in\mathcal{D}_3$
such that $H=I_{\psi_0}(G)$ indeed encodes a proper $3$-colouring of~$H_0$.
\end{proof}

\section{Questions and open problems}
Our interpretation approach and obtained results open several natural 
questions which we believe are worth further investigation.
We list them in this last section of the paper.
\begin{enumerate}[1.]
\item 
Can one characterize under which conditions on a formula 
$\psi(x,y)$ and a graph class $\ca C$, the following holds?
Given a graph $H \in I_\psi(\mathcal{C})$ as an input,
it would be possible to compute in polynomial 
(or in FPT with respect to $\psi$ and $\ca C$)
time a graph $G \in \mathcal{C}$ such that $H = \psi(G)$.
We know both of positive and negative examples
(Theorems~\ref{thm:decomposition} and~\ref{thm:recognizeinthard}),
but any plausible conjecture seems now out of reach.

\item 
It is easy to generalize the notion of near-$k$-twins $u,v$ in such a way that 
it would measure not the size of the symmetric difference between
the neighbourhoods, $|N(u)\symdiff N(v)|$,
but structural properties of the subgraph induced on $N(u)\symdiff N(v)$. 
For example, we may define a {\em near-$sd_k$-twin relation}, 
in which two vertices $u,v$ would be near-$sd_k$-twins 
if the subgraph induced on $N(u)\symdiff N(v)$
has shrub-depth at most $k$ (see \cite{ganianetal12} for the definition of
shrub-depth). 
One may then consider graph classes where the near-$sd_k$-twin relation is an equivalence. 
Is there an FPT algorithm for FO model checking on such graph classes?

\item \label{it:third}
Is it possible to extend our results to graph classes interpretable
in more general sparse graph classes?
For example, what is a characterization of graph classes interpretable in
trees or in planar graphs? In graph classes of bounded expansion?
Are there FPT algorithms for FO model checking on such classes?

\item
In relation to the previous point, we know from Corollary~\ref{cor:robust}
that the notion of near-uniform graph classes is robust under FO
interpretations and transductions.
We know of (at least) two other examples of such behaviour --
the graph classes of bounded clique-width~\cite{co00} and
the graph classes of bounded shrub-depth~\cite{ganianetal12}
(which are robust even under MSO transductions).
Can one come up with other natural and interesting graph properties
defining graph classes robust under FO transductions?

\item\label{it:fifth}
Inspired by the classification of sparse graph classes
by Ne\v{s}et\v{r}il and Ossona de Mendez \cite{NOdM12},
we may investigate graph classes $\ca D$ with the property that, for 
every FO formula $\psi(x,y)$ there exists a graph $F_\psi$
(as ``forbidden'') such that $F_\psi$ is not
present as an induced subgraph in any member of $I_\psi(\ca D)$.
This logical definition may be considered in analogy to the structural 
definition(s) of nowhere dense classes~\cite{NOdM12}
(as ``nowhere FO dense'').
What can we say about complexity of FO model checking on such classes
$\ca D$?
\end{enumerate}

To conclude, we make the following two explicit conjectures related to
points \ref{it:third} and \ref{it:fifth} of the discussion.
\begin{conj}
Let $\mathcal{C}$ be a nowhere dense graph class and $\mathcal{D}$ a 
graph class FO interpretable in $\mathcal{C}$. 
Then $\mathcal{D}$ has an FPT algorithm for FO model checking.
\end{conj}
\begin{conj}[``Nowhere FO dense'']\label{conj:noFOdense}
Let $\mathcal{D}$ be a graph class with the following property:
for every FO formula $\psi(x,y)$ there exists a graph $F_\psi$
such that $F_\psi$ is not an induced subgraph of any member of $I_\psi(\ca D)$.
Then $\mathcal{D}$ has an FPT algorithm for FO model checking.
\end{conj}

Regarding Conjecture~\ref{conj:noFOdense}, it is tempting to strengthen its
conclusion to; `then $\ca D$ is FO interpretable in some nowhere
dense graph class', but that actually fails.
For example, take the class $\mathcal{D}$ of graphs of clique-width~$2$.
Then, for every FO formula $\psi(x,y)$, the interpreted class 
$I_\psi(\ca D)$ is of bounded clique-width, too,
and so a forbidden graph $F_\psi$ always exists in this case.
However, the class $\mathcal{D}$ is not interpretable in any nowhere
dense graph class.


%
%
%


\bibliographystyle{abbrv}

\end{document}